\documentclass[a4paper,onecolumn,10pt,accepted=2025-08-12]{quantumarticle}

\pdfoutput=1

\usepackage[T1]{fontenc}
\usepackage{authblk}
\usepackage{amsfonts}
\usepackage{amsmath}
\usepackage{amssymb}
\usepackage{amsthm}
\usepackage{algorithm}
\usepackage{comment}
\usepackage{algpseudocode}
\usepackage{graphicx}
\usepackage[colorlinks=true,linkcolor=violet,citecolor=violet,urlcolor=violet,plainpages=false,pdfpagelabels]{hyperref}
\usepackage{color}
\usepackage[dvipsnames]{xcolor}
\usepackage{mathtools}
\usepackage{bbm}
\usepackage{float}
\usepackage{array}
\usepackage{verbatim}
\usepackage{enumitem}
\usepackage{tocloft}
\usepackage[nottoc]{tocbibind}
\usepackage[font=small,labelfont=sc,margin=1cm]{caption}
\usepackage{dsfont}
\usepackage[normalem]{ulem}
\usepackage{quantikz}
\usepackage{soul}
\usepackage{cleveref}

\usepackage{anyfontsize}  
\usepackage{microtype} 

\RequirePackage[backend=biber,backref=false,style=phys,arxiv=abs,eprint=true,articletitle=true,biblabel=brackets,chaptertitle=false,pageranges=false,bibencoding=utf8]{biblatex}  
\appto{\bibsetup}{\sloppy} 

\allowdisplaybreaks

\makeatletter
\g@addto@macro\bfseries{\boldmath}
\makeatother

\newcommand{\ketbra}[2]{\ket{#1}\!\bra{#2}}

\newcommand{\abs}[1]{|#1|}

\newcommand{\norm}[1]{\lVert#1\rVert}

\DeclareUnicodeCharacter{2286}{\subseteq} 

\theoremstyle{definition}
\newtheorem{theorem}{Theorem}
\newtheorem{lemma}[theorem]{Lemma}
\newtheorem{corollary}[theorem]{Corollary}
\newtheorem{definition}[theorem]{Definition}

\newtheorem{remark}[theorem]{Remark}

\newtheorem{claim}[theorem]{Claim}
\newtheorem{fact}[theorem]{Fact}

\renewcommand{\qedsymbol}{$\blacksquare$}

\renewcommand{\qedsymbol}{\unskip\nobreak\quad\qedsymbol}
\renewcommand{\qedsymbol}{$\blacksquare$}

\definecolor{dred}  {RGB}{164,12,52}

\newcommand{\poly}{{\mathrm{poly}}}

\renewcommand{\Pr}{\mathop{\rm Pr}\nolimits}

\newcommand{\ontop}[2]{{\begin{array}{l} {#1} \\ {#2} \end{array}}}
\newcommand{\stackket}[2]{{\Bigg | \hskip -5pt \ontop{#1}{#2} \hskip -3pt  \Bigg \rangle}}
\newcommand{\stackbra}[2]{{\Bigg \langle \hskip -5pt \ontop{#1}{#2} \hskip -3pt  \Bigg |}}
\newcommand{\stackketbra}[2]{{ \stackket{#1}{#2} \stackbra{#1}{#2} }}

\newlength{\bigcirclen}

\def\statea{{\bigcirc}}
\def\stateda{\begin{tikzpicture}[baseline={([yshift=-.5ex]current bounding box.center)},vertex/.style={anchor=base,
    circle,fill=black!25,minimum size=18pt,inner sep=2pt}]
    \node[draw, rectangle] (box) at (0, 0) {
        $\statea \statea$
    };
\end{tikzpicture}}
\def\stateba{{\settowidth{\bigcirclen}{$\bigcirc$}\makebox[0pt][l]{\makebox[\bigcirclen][c]{$\uparrow$}}\bigcirc}}
\def\encircledlflag{{\settowidth{\bigcirclen}{$\bigcirc$}\makebox[0pt][l]{\makebox[\bigcirclen][c]{$\twoheadleftarrow$}}\bigcirc}}
\def\encircledrflag{{\settowidth{\bigcirclen}{$\bigcirc$}\makebox[0pt][l]{\makebox[\bigcirclen][c]{$\Rrightarrow$}}\bigcirc}}
\def\statebhr{{\settowidth{\bigcirclen}{$\bigcirc$}\makebox[0pt][l]{\makebox[\bigcirclen][c]{$\rightarrow$}}\bigcirc}}

\def\statebb{{\settowidth{\bigcirclen}{$\bigcirc$}\makebox[0pt][l]{\makebox[\bigcirclen][c]{$\downarrow$}}\bigcirc}}
\def\stateb{{\settowidth{\bigcirclen}{$\bigcirc$}\makebox[0pt][l]{\makebox[\bigcirclen][c]{$\uparrow$}}\makebox[0pt][l]{\makebox[\bigcirclen][c]{$\downarrow$}}\bigcirc}}
\def\statebl{{\settowidth{\bigcirclen}{$\bigcirc$}\makebox[0pt][l]{\makebox[\bigcirclen][c]{$\leftarrow$}}\bigcirc}}
\def\statebr{{\settowidth{\bigcirclen}{$\bigcirc$}\makebox[0pt][l]{\makebox[\bigcirclen][c]{$\rightarrow$}}\bigcirc}}
\def\statebh{{\settowidth{\bigcirclen}{$\bigcirc$}\makebox[0pt][l]{\makebox[\bigcirclen][c]{$\rightarrow$}}\makebox[0pt][l]{\makebox[\bigcirclen][c]{$\leftarrow$}}\bigcirc}}
\def\stateca{{\settowidth{\bigcirclen}{$\bigcirc$}\makebox[0pt][l]{\makebox[\bigcirclen][c]{$\Uparrow$}}\bigcirc}}
\def\statecb{{\settowidth{\bigcirclen}{$\bigcirc$}\makebox[0pt][l]{\makebox[\bigcirclen][c]{$\Downarrow$}}\bigcirc}}
\def\statec{{\settowidth{\bigcirclen}{$\bigcirc$}\makebox[0pt][l]{\makebox[\bigcirclen][c]{\raisebox{.035em}{\fontsize{9}{0}\selectfont$\Updownarrow$}}}\bigcirc}}
\def\statehr{{\settowidth{\bigcirclen}{$\bigcirc$}\makebox[0pt][l]{\makebox[\bigcirclen][c]{$\Rightarrow$}}\bigcirc}}
\def\statehl{{\settowidth{\bigcirclen}{$\bigcirc$}\makebox[0pt][l]{\makebox[\bigcirclen][c]{$\Leftarrow$}}\bigcirc}}
\def\statehrl{{\settowidth{\bigcirclen}{$\bigcirc$}\makebox[0pt][l]{\makebox[\bigcirclen][c]{\raisebox{.035em}{\fontsize{9}{0}\selectfont$\Leftrightarrow$}}}\bigcirc}}

\def\dstatebb{
\begin{tikzpicture}[baseline={([yshift=-.5ex]current bounding box.center)},vertex/.style={anchor=base,
    circle,fill=black!25,minimum size=18pt,inner sep=2pt}]
    \node[draw, rectangle] (box) at (0, 0) {
        $\stateb \statebh$
    };
\end{tikzpicture}
}
\def\dstatecc{
\begin{tikzpicture}[baseline={([yshift=-.5ex]current bounding box.center)},vertex/.style={anchor=base,
    circle,fill=black!25,minimum size=18pt,inner sep=2pt}]
    \node[draw, rectangle] (box) at (0, 0) {
        $\statec \statehrl$
    };
\end{tikzpicture}
}

\def\dstatecb{
\begin{tikzpicture}[baseline={([yshift=-.5ex]current bounding box.center)},vertex/.style={anchor=base,
    circle,fill=black!25,minimum size=18pt,inner sep=2pt}]
    \node[draw, rectangle] (box) at (0, 0) {
        $\statec \statebh$
    };
\end{tikzpicture}
}

\def\dstated{
\begin{tikzpicture}[baseline={([yshift=-.5ex]current bounding box.center)},vertex/.style={anchor=base,
    circle,fill=black!25,minimum size=18pt,inner sep=2pt}]
    \node[draw, rectangle] (box) at (0, 0) {
        $\stated \stated$
    };
\end{tikzpicture}
}

\def\statecahr{
\begin{tikzpicture}[baseline={([yshift=-.5ex]current bounding box.center)},vertex/.style={anchor=base,
    circle,fill=black!25,minimum size=18pt,inner sep=2pt}]
    \node[draw, rectangle] (box) at (0, 0) {
        $\stateca \statehr$
    };
\end{tikzpicture}
}
\def\statecabhr{
\begin{tikzpicture}[baseline={([yshift=-.5ex]current bounding box.center)},vertex/.style={anchor=base,
    circle,fill=black!25,minimum size=18pt,inner sep=2pt}]
    \node[draw, rectangle] (box) at (0, 0) {
        $\stateca \statebhr$
    };
\end{tikzpicture}
}

\def\statebabhr{
\begin{tikzpicture}[baseline={([yshift=-.5ex]current bounding box.center)},vertex/.style={anchor=base,
    circle,fill=black!25,minimum size=18pt,inner sep=2pt}]
    \node[draw, rectangle] (box) at (0, 0) {
        $\stateba \statebhr$
    };
\end{tikzpicture}
}

\def\statecbhl{\begin{tikzpicture}[baseline={([yshift=-.5ex]current bounding box.center)},vertex/.style={anchor=base,
    circle,fill=black!25,minimum size=18pt,inner sep=2pt}]
    \node[draw, rectangle] (box) at (0, 0) {
        $\statecb \statehl$
    };
\end{tikzpicture}}

\def\statecbhbl{\begin{tikzpicture}[baseline={([yshift=-.5ex]current bounding box.center)},vertex/.style={anchor=base,
    circle,fill=black!25,minimum size=18pt,inner sep=2pt}]
    \node[draw, rectangle] (box) at (0, 0) {
        $\statecb \statebl$
    };
\end{tikzpicture}}

\def\statebbhl{
\begin{tikzpicture}[baseline={([yshift=-.5ex]current bounding box.center)},vertex/.style={anchor=base,
    circle,fill=black!25,minimum size=18pt,inner sep=2pt}]
    \node[draw, rectangle] (box) at (0, 0) {
        $\statebb \statebl$
    };
\end{tikzpicture}
}
\def\stated{{\settowidth{\bigcirclen}{$\bigcirc$}
\makebox[0pt][l]{\makebox[\bigcirclen][c]{$\times$}}\bigcirc}}

\def\fstateea{
\begin{tikzpicture}[baseline={([yshift=-.5ex]current bounding box.center)},vertex/.style={anchor=base,
    circle,fill=black!25,minimum size=18pt,inner sep=2pt}]
    \node[draw, rectangle] (box) at (0, 0) {
        $\stateelse \statea$
    };
\end{tikzpicture}
}
\def\fstateae{
\begin{tikzpicture}[baseline={([yshift=-.5ex]current bounding box.center)},vertex/.style={anchor=base,
    circle,fill=black!25,minimum size=18pt,inner sep=2pt}]
    \node[draw, rectangle] (box) at (0, 0) {
        $\statea \stateelse$
    };
\end{tikzpicture}
}
\def\stateelse{{\settowidth{\bigcirclen}{$\bigcirc$}\makebox[0pt][l]{\makebox[\bigcirclen][c]{$\bullet$}}\bigcirc}}
\def\fstateelse{
\begin{tikzpicture}[baseline={([yshift=-.5ex]current bounding box.center)},vertex/.style={anchor=base,
    circle,fill=black!25,minimum size=18pt,inner sep=2pt, opacity=0.2}]
    \node[draw, rectangle] (box) at (0, 0) {
        $\stateelse \stateelse$
    };
\end{tikzpicture}
}
\def\fstatede{
\begin{tikzpicture}[baseline={([yshift=-.5ex]current bounding box.center)},vertex/.style={anchor=base,
    circle,fill=black!25,minimum size=18pt,inner sep=2pt}]
    \node[draw, rectangle] (box) at (0, 0) {
        $\stated \stateelse$
    };
\end{tikzpicture}
}
\def\fstateed{
\begin{tikzpicture}[baseline={([yshift=-.5ex]current bounding box.center)},vertex/.style={anchor=base,
    circle,fill=black!25,minimum size=18pt,inner sep=2pt}]
    \node[draw, rectangle] (box) at (0, 0) {
        $\stateelse \stated$
    };
\end{tikzpicture}
}
\def\fdstatecb{
\begin{tikzpicture}[baseline={([yshift=-.5ex]current bounding box.center)},vertex/.style={anchor=base,
    circle,fill=black!25,minimum size=18pt,inner sep=2pt}]
    \node[draw, rectangle] (box) at (0, 0) {
        $\stateb \statehrl$
    };
\end{tikzpicture}
}

\def\turnflag{
\begin{tikzpicture}[baseline={([yshift=-.5ex]current bounding box.center)},vertex/.style={anchor=base,
    circle,fill=black!25,minimum size=18pt,inner sep=2pt}]
    \node[draw, rectangle] (box) at (0, 0) {
        $\hspace{5.5pt}\circlearrowleft \hspace{5.5pt}$
    };
\end{tikzpicture}
}
\def\statebhboxed{
\begin{tikzpicture}[baseline={([yshift=-.5ex]current bounding box.center)},vertex/.style={anchor=base,
    circle,fill=black!25,minimum size=18pt,inner sep=2pt}]
    \node[draw, rectangle] (box) at (0, 0) {
        $\hspace{5.5pt} \statebh \hspace{5.5pt}$
    };
\end{tikzpicture}
}

\def\statehrlboxed{
\begin{tikzpicture}[baseline={([yshift=-.5ex]current bounding box.center)},vertex/.style={anchor=base,
    circle,fill=black!25,minimum size=18pt,inner sep=2pt}]
    \node[draw, rectangle] (box) at (0, 0) {
        $\hspace{5.5pt} \statehrl\hspace{5.5pt}$
    };
\end{tikzpicture}
}
\def\statedboxed{
\begin{tikzpicture}[baseline={([yshift=-.5ex]current bounding box.center)},vertex/.style={anchor=base,
    circle,fill=black!25,minimum size=18pt,inner sep=2pt}]
    \node[draw, rectangle] (box) at (0, 0) {
        $\hspace{5.5pt}\stated\hspace{5.5pt}$
    };
\end{tikzpicture}
}

\def\stateaboxed{
\begin{tikzpicture}[baseline={([yshift=-.5ex]current bounding box.center)},vertex/.style={anchor=base,
    circle,fill=black!25,minimum size=18pt,inner sep=2pt}]
    \node[draw, rectangle] (box) at (0, 0) {
        $\hspace{5.5pt}\statea\hspace{5.5pt}$
    };
\end{tikzpicture}
}
\def\statecboxed{
\begin{tikzpicture}[baseline={([yshift=-.5ex]current bounding box.center)},vertex/.style={anchor=base,
    circle,fill=black!25,minimum size=18pt,inner sep=2pt}]
    \node[draw, rectangle] (box) at (0, 0) {
        $\hspace{5.5pt}\statec\hspace{5.5pt}$
    };
\end{tikzpicture}
}
\def\statebboxed{
\begin{tikzpicture}[baseline={([yshift=-.5ex]current bounding box.center)},vertex/.style={anchor=base,
    circle,fill=black!25,minimum size=18pt,inner sep=2pt}]
    \node[draw, rectangle] (box) at (0, 0) {
        $\hspace{5.5pt}\stateb\hspace{5.5pt}$
    };
\end{tikzpicture}
}

\def\lflag{
\begin{tikzpicture}[baseline={([yshift=-.5ex]current bounding box.center)},vertex/.style={anchor=base,
    circle,fill=black!25,minimum size=18pt,inner sep=2pt}]
    \node[draw, rectangle] (box) at (0, 0) {
        $\hspace{5.5pt} \encircledlflag \hspace{5.5pt}$
    };
\end{tikzpicture}
}
\def\rflag{
\begin{tikzpicture}[baseline={([yshift=-.5ex]current bounding box.center)},vertex/.style={anchor=base,
    circle,fill=black!25,minimum size=18pt,inner sep=2pt}]
    \node[draw, rectangle] (box) at (0, 0) {
        $\hspace{5.5pt} \encircledrflag \hspace{5.5pt}$
    };
\end{tikzpicture}
}
\def\stateelseboxed{
\begin{tikzpicture}[baseline={([yshift=-.5ex]current bounding box.center)},vertex/.style={anchor=base,
    circle,fill=black!25,minimum size=18pt,inner sep=2pt}]
    \node[draw, rectangle] (box) at (0, 0) {
        $\hspace{5.5pt} \stateelse \hspace{5.5pt}$
    };
\end{tikzpicture}
}

\title{Complexity of geometrically local stoquastic Hamiltonians}

\begin{document}

\author{Asad Raza}

\affiliation{Dahlem Center for Complex Quantum Systems, Freie Universit\"at Berlin, 14195 Berlin, Germany}

\affiliation{Sorbonne Universit{\'e}, CNRS, LIP6, 75005 Paris, France}

\author{Jens Eisert}

\affiliation{Dahlem Center for Complex Quantum Systems, Freie Universit\"at Berlin, 14195 Berlin, Germany}

\affiliation{Helmholtz-Zentrum Berlin f{\"u}r Materialien und Energie, 14109 Berlin, Germany}

\author{Alex B. Grilo}

\affiliation{Sorbonne Universit{\'e}, CNRS, LIP6, 75005 Paris, France}
\maketitle

\begin{abstract}
The QMA-completeness of the local Hamiltonian problem is a landmark result of the field of Hamiltonian complexity that studies the computational complexity of problems in quantum many-body physics. Since its proposal, substantial effort has been invested in better understanding the problem for physically motivated important families of Hamiltonians. In particular, the QMA-completeness of approximating the ground state energy of local Hamiltonians has been extended to the case where the Hamiltonians are geometrically local in one and two spatial dimensions. Among those physically motivated Hamiltonians, stoquastic Hamiltonians play a particularly crucial role, as they constitute the manifestly sign-free Hamiltonians in Monte Carlo approaches. Interestingly, for such Hamiltonians, the problem at hand becomes more ``classical'', being hard for the class MA (the randomized version of NP) and its complexity has tight connections with derandomization. In this work, we prove that both the two- and one-dimensional geometrically local analogues remain MA-hard with high enough qudit dimension. Moreover, we show that related problems are StoqMA-complete.
\end{abstract}

\section{Introduction}

Understanding the properties of the ground state of a Hamiltonian is a fundamental problem in physics. Indeed, much of the low temperature physics can be derived from such ground state properties of interacting Hamiltonians. While a number of approximate methods are known for this purpose and the study of ground states constitutes a large field of research within condensed matter physics,
this is widely believed to be a computationally hard task in general. 
Kitaev \cite{KitaevClassicalQuantumComputation} formalized this belief, by proving the hardness of estimating the energy of the ground state of a Hamiltonian. More formally, Kitaev defined the \emph{local Hamiltonian problem} (LH-MIN), which asks to decide whether a local Hamiltonian on $n$ qubits 
(or particles with local dimension\footnote{Following previous results, we interchangeably use the terms {\em dimension $d$} and {\em $d$-state} to denote the local dimension of particles.} $d>2$) has ground state energy less than some real parameter $\alpha$ or
greater than a parameter $\beta$.  Kitaev showed that the LH-MIN problem, for
some $\beta - \alpha \geq {1}/{\rm poly}(n)$,  is QMA-complete, where
QMA is a quantum version of the complexity class NP. However, the Hamiltonians that arise in such a reduction are
rather artificial and,  in order to address this issue, the result of Kitaev has been
extended to several physically relevant Hamiltonians
\cite{kempe-kitaev-regev-2-local,kempe-regev-3-local,NagajMozes,PhysRevA.76.030307,Oliviera-Terhal-geometric-2D-QMA,Aharonov_2009}, in particular to Hamiltonians that are geometrically local, i.e., their particles lie in a lattice.
Such geometrically local Hamiltonians are ubiquitous in physics and make up essentially all models for
interacting quantum matter. A conclusion to be drawn from this is that fault tolerant quantum computer should not be able to efficiently estimate the ground state energy of local Hamiltonians to a high precision, and there is no 
superpolynomial quantum advantage \cite{RevModPhys.95.035001},
which may be seen as a remarkable insight in its own right. This is an example of the many results at the intersection of  theoretical computer science and quantum many-body theory, in particular it lies in the interdisciplinary field of Hamiltonian complexity \cite{HamiltonianComplexitySev,HamiltonianComplexityTobias} that studies the computational complexity of Hamiltonians arising in the context of condensed matter physics
\cite{HamiltonianComplexitySev,HamiltonianComplexityTobias}.

While we have such hardness results for general Hamiltonians, it is common that in physical systems such Hamiltonians feature
further structure or are equipped with additional notable features 
or symmetries. 
It then makes a lot of sense to ask if there exists special classes of Hamiltonians that admit characterizations in restricted complexity classes. 
A particularly important type of Hamiltonians 
is constituted by so-called \emph{stoquastic Hamiltonians}, which are Hamiltonians whose off-diagonals in the computational basis are non-positive.

While this feature may at first look like a detail,
it is not: stoquastic Hamiltonians are Hamiltonians that
avoid the so-called Monte Carlo sign problem. If Hamiltonians
are non-stoquastic and hence feature positive terms on the off-diagonals,
Monte Carlo sampling methods --- so methods that are a workhorse of the study
of quantum many-body systems 
\cite{PhysRevB.26.5033,TrebstMonteCarlo,PolletMonteCarlo} --- come with an overhead that scales
exponentially in the system size \cite{PhysRevLett.94.170201,SignEasing}. In this sense, stoquastic Hamiltonians are considered the ``easy Hamiltonians''
from the perspective of quantum Monte Carlo simulations.
A special case of such stoquastic Hamiltonians consists of terms where all off-diagonal terms 
vanishing, which corresponds to the classical $k$-MAXSAT 
problem, which is
NP-complete.

The seminal result of Bravyi and Terhal~\cite{bravyi2010complexity} settled
the complexity of \emph{stoquastic LH-MIN problem}, showing that it is
MA-complete.\footnote{The
name MA abbreviates Merlin-Arthur, and it corresponds to the probabilistic
version of NP.  More concretely, we have a randomized algorithm that verifies an
untrusted proof and for yes-instances, there must exist a proof such that the
verification algorithm accepts with high probability, whereas for no-instances,
all proofs should lead to low-probability acceptance.}
We remark that the complexity class MA was introduced by
Babai \cite{MA} in the 80s and the result of Bravyi and Terhal
\cite{bravyi2010complexity} showed the very first natural MA-complete problem. It is interesting to
 see 
 that this gives rise to a 
class of quantum Hamiltonians 
that do not require `quantum resources'
for the verification of the 
ground state energy. Instead, a mere classical verification suffices to decide between the upper and lower bounds on the ground state energy in this reading of the local Hamiltonian problem.

This (classical) MA-completeness appears to be even more striking when one notes that stoquastic local Hamiltonians capture inherently quantum processes such as ferromagnetic Heisenberg models and the quantum transverse Ising model, among others, so models that feature strongly in condensed matter and statistical physics. In some ways, stoquastic Hamiltonians 
can be seen as being borderline classical. At the same time, a superpolynomial oracle separation has been shown between the power of adiabatic quantum computation with no sign problem and the power of classical computation
\cite{gilyen2021sub}. For all of the above reasons, one can reasonably argue that stoquasticity is critical to understanding the distinction between quantum and
classical complexity. 

We also note that  Aharonov and
Grilo~\cite{aharonov2019stoquastic} have proven that
the stoquastic LH-MIN problem for $\alpha = 0$ and $\beta = \Omega(1)$ is in NP.
This shows a surprising connection between gap amplification for stoquastic Hamiltonians,
which is a restricted version of quantum PCPs~\cite{aharonov2019stoquastic}, 
and the derandomization conjecture that 
MA = NP.  While this gives a potential avenue to study such an important problem in classical complexity theory, 
a big difficulty in such direction is the small number of
MA-complete problems~\cite{aharonov2021two, jiang2023local, bravyiTIM, bravyi2006complexity, bravyi2006merlinarthur}, compared to the plethora of NP-complete problems. In particular, having more {\em structured}
MA-complete problems could give us new handles to solve open problems involving stoquastic Hamiltonians.

With this motivation in mind, we ask whether the
MA-hardness of the algebraically local stoquastic Hamiltonian problem extends to Hamiltonians that are geometrically local---and hence
relate closely to structured Hamiltonians that would naturally be encountered
in the physical context. After all, interacting terms almost
exclusively act in such a geometrically local fashion in
physically plausible models; in fact, 
ubiquitous are terms that are not only geometrically
local but that involve nearest neighbours only. We notice that the inclusion in
MA follows directly from the result of Bravyi and Terhal \cite{bravyi2010complexity}, and the remaining
question is if these Hamiltonians are also MA-hard.

In this work, we 
answer this question to the affirmative.
Concretely, we
provide two separate constructions (for
models in one and two spatial dimensions) that show that geometrically-local
stoquastic LH-MIN problem  is MA-hard.
Interestingly, even in this
substantially more physically plausible reading of the 
stoquastic local 
Hamiltonian problem, MA-hardness is retained.
This work hence contributes to the broader research question of distinguishing quantum from classical complexity of Hamiltonians, in a context that 
is provided relevance because of the crucially important role  stoquasticity plays in Monte Carlo sampling methods \cite{PhysRevLett.94.170201,PolletMonteCarlo,CuringSignProblem,SignEasing}.

\subsection{Our results and technical contributions}
\label{sec:intro-technical}

The main contribution of this work is to prove 
MA-hardness of geometrically local stoquastic Hamiltonians. We stress again that most
Hamiltonians that are physically plausible feature such
short-ranged (mostly nearest neighbour) and hence geometrically local 
interacting terms. Models of quantum materials 
and toy models in condensed matter physics 
in particular are almost exclusively of this form.
For this reason, studying geometrically local models has been
a key aim early on in the field of Hamiltonian complexity.
Our first contribution is to prove the MA-hardness of stoquastic local Hamiltonians on a two-dimensional lattice. We note that Bravyi and Terhal \cite{bravyi2010complexity} already showed that if a Hamiltonian, in addition to being stoquastic, is also frustration-free (meaning that the ground state of the full Hamiltonian is also the simultaneous ground state of each of the local Hamiltonian terms), then the LH-MIN problem for such Hamiltonians is also contained in MA. Using this inclusion in MA due to Bravyi and Terhal \cite{bravyi2010complexity}, we state following two theorems establishing MA-completeness of stoquastic {\em frustration-free} Hamiltonians. 

\begin{theorem}[MA-completeness of the 
geometrically local, frustration-free stoquastic Hamiltonian problem on a two-dimensional lattice]\label{theorem:2d-ma-hard}
The frustration-free stoquastic local Hamiltonian problem with 14-dimensional qudits on a two-dimensional lattice with geometrically two-local interactions is MA-complete.
\end{theorem}
We shall give the proof of Theorem~\ref{theorem:2d-ma-hard} over two sections: Completeness (Section \ref{subsection:completeness-proof}) and soundness (Section \ref{subsection:soundness-proof}). 

Then, we follow up and consider MA-completeness of stoquastic local Hamiltonians on a one-dimensional lattice.

\begin{theorem}[MA-completeness of the  geometrically local, frustration-free stoquastic Hamiltonian problem on a line]\label{theorem:1D-MA-hardness}
The frustration-free stoquastic local Hamiltonian problem with 19-dimensional qudits on a one-dimensional line with geometrically two-local interactions is MA-complete.
\end{theorem}

Ref.~\cite{bravyi2006merlinarthur} has defined another class, StoqMA, that lies between MA and QMA: MA $\subseteq$ StoqMA $\subseteq$ QMA. Authors in Ref.~\cite{bravyi2006merlinarthur} showed that deciding the general stoquastic local Hamiltonian problem (without frustration-free assumption) with arbitrary $\alpha$ and $\beta$ such that $\beta - \alpha \geq \frac{1}{\poly(n)}$ is StoqMA-complete using a slight variation of the standard circuit-to-Hamiltonian construction. We can also apply a similar idea to our constructions, leading to StoqMA-completeness of geometrically local stoquastic Hamiltonian problem.

\begin{theorem}[StoqMA-completeness  of the 
geometrically local stoquastic Hamiltonian problem]
The stoquastic local Hamiltonian problem with 14-state qudits on a two-dimensional lattice with geometrically two-local interactions is StoqMA-complete. Moreover,
the geometrically two-local stoquastic LH-MIN also stays StoqMA-complete on a one-dimensional line with 19-dimensional particles. See Corollary~\ref{corollary-stoqma-completeness} in Section~\ref{section:stoqma-completeness}.
\end{theorem}

\subsubsection{Technical contributions}

In this subsection, we discuss the technical challenges to achieve our results and how we handle them.
The first geometric construction was introduced in 
Ref.~\cite{aharonov2008adiabatic} for implementing unitary circuits using local Hamiltonians on a
two-dimensional grid with six-dimensional qudits. Subsequently, QMA-completeness on a 
one-dimensional lattice has also been proven to be true for qudits with 
twelve local dimensions \cite{Aharonov_2009} (the 
local particle dimension has later been brought down to 
eight \cite{nagaj-1d}). Moreover, the locality of the two-dimensional case was also reduced to qubits by Oliveira and Terhal \cite{Oliviera-Terhal-geometric-2D-QMA}.
One of the main technical challenges in any geometric construction is to ensure that the Hamiltonian remains geometrically local --- ideally two-local reflecting nearest neighbour interactions --- while respecting the inverse polynomial spectral gap bounds needed to prove completeness and soundness conditions. 
In our case, we must also ensure that each Hamiltonian term $H_j$ in Eq.~\eqref{eq:full-hamilt} is at the same time stoquastic. 

Unlike QMA where the the verification uses arbitrary quantum gates, to prove MA-hardness, we only need to consider classically reversible universal gate set. In particular, we can consider only Toffoli and $X$ gates. While they will help us ensure stoquasticity for each Hamiltonian term, given that they will appear as $-U$ on the off-diagonal blocks in the reduction, using Toffolis comes at a notable cost. Since a Toffoli gate acts on three qubits, the challenge is to realize this gate using only two-local Hamiltonians. Directly applying the existing constructions~\cite{aharonov2008adiabatic,Aharonov_2009} will result in a three-local Hamiltonian.  We overcome this by increasing particle dimension (and thus reducing locality) to be able to handle three logical qubits at a time and ensure that all the terms are only nearest neighbour (two-local) while also being stoquastic.
Needless to say, there are other summands in our Hamiltonian. However, in our construction they are either projectors or matrices that are signed (negative) outer product between two computational basis vectors; both of these are stoquastic. Our goal will be to ensure that each $H_j$ is stoquastic, which  will in turn ensure that the full Hamiltonian
\begin{equation}\label{eq:full-hamilt}
H= \sum_{j=1}^r H_j
\end{equation}
is also stoquastic. The Hamiltonian in Eq.~\eqref{eq:full-hamilt} is defined on $n$
local systems, referred to as ``particles'', of dimension $d$,
called qubits for $d=2$.

Our main technical contribution is twofold: to map the MA verification circuit onto (i) a two-dimensional grid of qudits (\Cref{theorem:2d-ma-hard}),  and (ii) on a one-dimensional line (\Cref{theorem:1D-MA-hardness}). To do so, we uniquely map every state in the circuit verification procedure to an arrangement on a grid (in the two-dimensional case) and later on a line (in the one-dimensional case). On a two-dimensional lattice, the Hilbert space of each particle is represented as 
\begin{equation}\label{eq-Hilbert-space-decomp-2}
\mathcal{H}_{\rm Particle} = \stateda \oplus \dstated \oplus \dstatebb \oplus \dstatecb \oplus \dstatecc, 
\end{equation}
where each of the last three rectangular boxes can be interpreted as the tensor product Hilbert space of two qubits --- which are the circles within the box (See \Cref{subsection:clock-prop} for a detailed explanation.) The first two boxes are simply one-dimensional spaces to keep track of how far we are in the computation and do not carry computational data. The way we implement the three-qubit Toffoli gate in a two-local way is essentially by acting the gate on two of these nearest neighbour particles. The Toffoli gate acts on only eight local dimensions between two neighbouring particles at a time. The application of a Toffoli gate is tracked by changing the single line (denoting that the unitary never touched the qubit) to a double line (denoting that unitary touched the qubit at least once). For example, 
\begin{equation}\label{eq-2D-unitary-action-example}
U_{1,2,3}\left(\stackket{\dstatecc}{\dstatebb} \right) =  \stackket{\dstatecc}{\dstatecb},
\end{equation}
where $U_{1,2,3}$ acts on three qubits numbered $1, 2$ and $3$ (the top two and the bottom left), but only on two nearest neighbour particles at a time. The top two particles already being double-lined denotes that those have been acted upon by a previous (Toffoli) gate. And that the new Toffoli gate acts on one extra untouched particle, i.e., the one on the bottom left.

The overarching idea in the one-dimensional construction also hinges upon this observation of encoding qubits in higher dimensional qudits and then applying Toffolis (via local Hamiltonians) in a two-local way on those higher dimensional particles. The difficulty in the one-dimensional case, however, is that we no longer have a possibility of nearest neighbour interactions across both the horizontal and the vertical, as we had in two spatial dimensions. In the two-dimensional setting, the nearest neighbour interactions across the vertical are used to track gate application, while the nearest neighbour horizontal ones help to initialize the particles in the following column in a state such that they are ready to be acted upon by the new layer of unitaries. Consequently, in the one-dimensional construction, we can no longer directly initialize the next block of particles in a desired state using only two-local interactions. We remedy this by taking polynomially (in the number of qubits of the verification circuit) many additional steps to reinitialize them. 

The upshot, similar to the one-dimensional QMA-hardness construction \cite{Aharonov_2009}, is that we now encounter states that are forbidden configurations and not locally checkable. In the two-dimensional construction, we might also have forbidden arrangements, but we can penalize them using only two local Hamiltonian terms. In the one-dimensional setting, we can no longer detect all the forbidden configurations in a two-local way. This has a direct consequence on the soundness constraint of our problem. We show that the Hilbert space of forbidden states that are not locally checkable admits the same lower bound as the legal states (\Cref{lemma:illegal-non-local}). Note that this does not affect completeness at all because it is merely an existence statement, for which we use the history state of our verification as this state which is a legal state and is entirely unaffected by local (or not local) checks on the forbidden states. 
Finally, we show in Theorem~\ref{theorem:1D-MA-hardness}
the respective hardness on a line.

\subsection{Preliminaries}

In our argument, we build on and further develop a technical machinery that has been developed over the years in the field of Hamiltonian complexity \cite{HamiltonianComplexityTobias,HamiltonianComplexitySev}. Subsequently, we will state a number of important concepts in technical terms.

\begin{definition}[Algebraically local Hamiltonian]\label{definition:k-local-hamiltonian}
A Hamiltonian $H = \sum_{j=1}^r H_j$ is $k$-local on $n$
particles when each $H_j$ acts non-trivially on at most any $k = O(1)$ out of $n$ particles.
\end{definition}

We notice that sometimes in the (computer science) literature, such algebraically local Hamiltonians are also called {\em combinatorially} local Hamiltonians.

\begin{definition}[Local Hamiltonian problem (LH-MIN$_{a,b}$)
]\label{definition:LH-MIN}
For an $n$-qudit Hamiltonian 
$H = \sum_{j=1}^r H_j$, 
where each $0 \preccurlyeq H_j \preccurlyeq \mathbbm{1}$ and each $H_j$ acts non-trivially on $k$ out of $n$ particles, decide whether $\lambda_{\min}(H) \leq a$ or $\lambda_{\min}(H) \geq b$.
\end{definition}

In our case, these particles shall have dimensions $d>2$ and each $H_j$ shall only act on nearest neighbour constituents.
We will call the local quantum systems \emph{particles} and will resort to the sites of the lattice associated with quantum degrees of freedom as \emph{constituents}. 

\begin{definition}[Geometrically local Hamiltonian problem]\label{definition:geometrically-local-LH-MIN}
   The geometrically $k$-local Hamiltonian problem is defined in the same way as its algebraic counterpart (Definition~\ref{definition:k-local-hamiltonian}), except that each Hamiltonian term acts on at most $k$ qubits that are geometric neighbors in the lattice.
\end{definition}
\begin{definition}[NP]\label{definition:NP}
A language $L$ is in NP iff there exists a deterministic polynomial time verification circuit, $C$, such that for every $x \in L_{\rm Yes}$ there exists a witness $w \in \{0, 1\}^{p(|x|)}$ for some polynomial $p$ for which $C(x, w) = 1$ (the verifier accepts). If $x \in L_{\rm No}$, then for all $w \in \{0, 1\}^{p(|x|)}$, $C(x, w) = 0$ (the verifier rejects). 
\end{definition}

We are now in the position to define the complexity class QMA, which is the quantum analogue of NP. Similar to the definition of NP, {\em acceptance} is defined when $\ket{1}$
is measured on the output qubit.

\begin{definition}[QMA]\label{definition:QMA}
A promise problem $L = (L_{\rm Yes}, L_{\rm No})$ is in QMA iff there exist polynomials $p$ and $q$ and an uniform family of quantum circuits $\{U_n\}_n$  such that $U_n$ has $\poly(n)$ size and for every $x \in L_{\rm Yes}$, there exists a witness as a quantum state vector $\ket{\psi}  \in (\mathbbm{C}^2)^{\otimes p(|x|)}$, such that $U_{|x|}(\ket{x}\ket{\psi}\ket{0}^{q(|x|)})$ accepts with probability  at least $2/3$  and if  $x \in L_{\rm No}$, then $\forall \ket{\psi}$, $U_{|x|}(\ket{x}\ket{\psi}\ket{0}^{q(|x|)})$ accepts with probability at most $1/3$.
\end{definition}

Kitaev~\cite{KitaevClassicalQuantumComputation} has shown that there exists some $a,b$ where $b - a = {1}/{\poly(n)}$ such that LH-MIN$_{a,b}$ is QMA-complete. Since then, researchers have tried to find special cases of the LH-MIN problem and characterize their complexity, largely motivated by understanding the impact of physically plausible constraints to the computational 
complexity of the problem. One such case is that of stoquastic Hamiltonians. 

\begin{definition}[Stoquasticity]\label{definition:Stoquasticity}
    A Hamiltonian $H = \sum_{j=1}^r H_j$ is called stoquastic with respect to a basis iff for every $j \in \{1,\dots ,r\}$ and $i \neq j \in \{0, 1\}^n$, $\bra{i}H_j\ket{j} \leq 0$. 
\end{definition}

We can then define the LH-MIN problem where we restrict all the terms to be stoquastic.

\begin{definition}[Stoquastic local Hamiltonian problem (stoquastic LH-MIN)]\label{definition:Stoq-LHMIN}
The stoquastic local Hamiltonian problem is defined in the same way as the local Hamiltonian problem (Definition~\ref{definition:LH-MIN}), except each $H_j$ in the full Hamiltonian $H = \sum_{j=1}^r H_j$ is term-wise stoquastic --- having real non-positive off-diagonals in the computational basis. 
\end{definition}

 Bravyi, DiVincenzo, Oliveira and Terhal \cite{bravyi2006complexity} have proven that algebraically local stoquastic LH-MIN is MA-hard. In order to understand this better, let us recall the definition of classes MA, which is a randomized version of NP.

\begin{definition}[MA]\label{definition:MA}
A promise problem $L = (L_{\rm Yes}, L_{\rm No})$
is in MA
iff there exist polynomials $p,q$ and an uniform family of (classical) circuits $\{C_n\}_n$  such that $C_n$ has $\poly(n)$ size and such that for every $x \in L_{\rm Yes}$, there exists a witness $w \in \{0, 1\}^{p(|x|)}$ such that $\Pr_{r \in \{0,1\}^{q(|x|)}}[C_{|x|}(x, w,r) \text{ accepts }] \geq \frac{2}{3}$. If $x \in L_{\rm No}$, then for all $w \in \{0, 1\}^{p(|x|)}$, $\Pr_{r \in \{0,1\}^{q(|x|)}}[C_{|x|}(x, w,r) \text{ accepts }] \leq \frac{1}{3}$.
\end{definition}

In this work, it is also useful to consider the coherent version of MA, that is called MA$_{\text{q}}$.

\begin{definition}[MA$_{\text{q}}$]\label{definition:MA_q}
A promise problem $L = (L_{\rm Yes}, L_{\rm No})$ 
is in MA$_{\text{q}}$
iff there exist polynomials $p,q,k$ and an uniform family of quantum circuits $\{U_n\}_n$  such that $U_n$ has $\poly(n)$ size and is composed only of Toffoli and X gates, and such that for every $x \in L_{\rm Yes}$, there exists a quantum state vector $\ket{\psi}  \in (\mathbbm{C}^2)^{\otimes p(|x|)}$ such that on the output qubit,  $U_{|x|}(\ket{x}\ket{\psi}\ket{0}^{q(|x|)}\ket{+}^{k(|x|)})$ accepts with probability  at least $2/3$. If $x \in L_{\rm No}$, then $\forall \ket{\psi}$, $U_{|x|}(\ket{x}\ket{\psi}\ket{0}^{q(|x|)}\ket{+}^{k(|x|)})$ accepts with probability at most $1/3$ on the output qubit.
\end{definition}

While it is well-known in classical complexity theory that the class MA equals its one-sided analogue, i.e., MA = MA$_{\text{1}}$ \cite{zachos87, goldreich_another_nodate}, the authors in 
Ref.\ \cite{bravyi2006complexity} have established another equivalence: MA and MA$_{\text{q}}$\footnote{Since the two classes are equivalent, we shall routinely switch between MA and MA$_{\text{q}}$, depending on whichever makes proofs convenient.}, and used that for the direct application of Kitaev's circuit-to-Hamiltonian mapping to prove MA-hardness of algebraically $k$-local stoquastic LH-MIN~\cite{bravyi2006complexity}. However, it is still an open question if stoquastic LH-MIN$_{a,b}$ is contained in MA for arbitrary $a,b$ where $b - a\geq {1}/{\poly(n)}$. Bravyi and Terhal~\cite{bravyi2006complexity} have shown the containment of LH-MIN$_{0,1/\poly(n)}$~\footnote{This problem is usually known as the frustration-free stoquastic Hamiltonian problem.} in MA.

Finally, we define StoqMA, a slight variant of MA$_{\text{q}}$.

\begin{definition}[StoqMA$_{a,b}$] StoqMA is MA$_{\text{q}}$, except that the verification algorithm accepts if the output qubit is in $\ket{+}$ and we consider completeness parameter $a$ and soundness parameter $b$.
\end{definition}

It has been shown that the problem stoquastic
LH-MIN$_{a,b}$ is contained in StoqMA$_{\alpha,\beta}$ and is StoqMA$_{\alpha',\beta'}$-hard, for some $\alpha - \beta \geq {1}/{\poly(n)}$ and 
$\alpha' - \beta' \geq 
{1}/{\poly(n)}$.
We remark that currently, it is not known how to reduce the error probability for StoqMA (if we did, we would have that StoqMA = MA ~\cite{aharonov2020stoqma}), and therefore one has to be careful how the parameters translate on the hardness and containment proofs.

\subsection{Related works}

\paragraph{Hardness of curing and easing
the sign problem.}
It is important to stress that the Hamiltonian that is the input to the stoquastic geometrically local Hamiltonian problem is already manifestly in a stoquastic form. This definition singles out the computational basis in which the off-diagonal terms of the Hamiltonian are 
non-positive. This particular basis in the definition
has no significance, and any other basis can be chosen as well. But as it turns out, for a Hamiltonian with real entries,
the computational task of finding local basis changes,
reflected by local orthogonal conjugations 
\begin{equation}
H\mapsto O^{\otimes n} H (O^\dagger)^{\otimes n} 
\end{equation}
for $O\in O(d)$,
so that the Hamiltonian is stoquastic is a computationally hard problem in its own right. Specifically, for multi-qubit two-local Hamiltonians that  contains one-local terms, then this task can be NP-hard \cite{CuringSignProblem}. For meaningful notions of easing the sign problem that relate to the overhead in Monte Carlo sample complexity, the problem remains an NP-complete task for nearest-neighbour Hamiltonians by a polynomial reduction to the MAXCUT-problem \cite{SignEasing}.

This adds an interesting twist to the problem at hand, as Hamiltonians are commonly not given in a manifestly stoquastic form, but in terms of some decomposition into local terms. So if given some form of a Hamiltonian, the very task that decides whether it can be brought into a stoquastic form in the computational basis as
being defined in the input of the geometrically local stoquastic 
Hamiltonian problem is computationally hard.

\section{MA-hardness of 2-local 14-state stoquastic Hamiltonians in two dimensions}\label{sec:2d-ma-hardness}

In this section, we introduce and describe our construction for 14-state 2-local Hamiltonian on a two-dimensional lattice that is MA-hard. It follows closely the construction of Ref.~\cite{Aharonov_2009}, but, as we discussed in Section~\ref{sec:intro-technical}, needs to handle three-qubit Toffoli gates. This requires a new geometric construction that uniquely maps the MA$_{\text{q}}$ verification circuit to the shape of the grid. To do so, we will first modify the MA$_{\text{q}}$ circuit, without loss of generality, to the one shown in Fig.~\ref{fig:circuit-construction}. Then we encode the execution and consequently acceptance/rejection of this circuit into Hamiltonian terms. In doing so, we will ensure that the each Hamiltonian term stays geometrically two-local and stoquastic. 

\subsection{Structure of verification circuit}
\label{sec:layers}
In our circuit-to-Hamiltonian mapping, we will need more structure on the layout of the circuit (as in Ref.\ \cite{Aharonov_2009}).
In particular, we require the $n'$-qubit verification circuit to be composed of layers, where we alternate a {\em computational layer} and an {\em identity layer}. In each {\em computational} layer, the first two 1- and 2- qubit gates (respectively) are chosen to be identity gates in order to make the exposition simpler. The remaining $n'-2$ gates in the computational layer are $3$-qubit gates, where the $i$-th gate acts on qubits $i,i+1,i+2$ for $i \in [n' - 2]$. 
Since we exclusively work with MA$_{\text{q}}$ verification, we are limited to classical computation (for example, we can consider the case where we have the product of X and Toffoli gates). In the identity layer, we have $n'/2$ identity gates, where the $i$-th gate of the layer acts on qudits $2i-1, 2i$, where $i \in [n'/2]$. We depict one computational and one identity layer in Figure~\ref{fig:circuit-construction}. 

We notice that we can assume that the MA$_{\text{q}}$ verification circuit has this structure by increasing its size by only a polynomial factor. We assume that the circuit has an even number of qubits (and we can add an extra dummy qubit if this is not the case).

\begin{figure}[ht]
    \centering
    \includegraphics[scale=0.5]{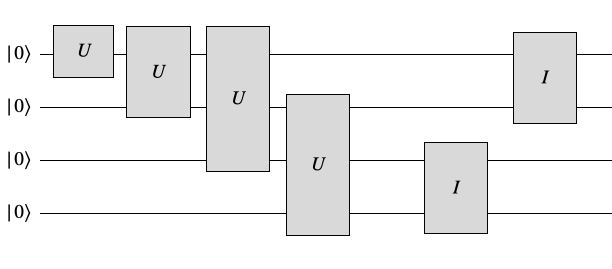}
    \caption{Modified circuit for our circuit-to-Hamiltonian mapping. We first make sure that the number of qubits in our verification circuit is even  (by adding an auxiliary, if needed). This ensures we can apply $n'/2$ many 2-qubit identity gates. Each unitary $U$ here is either a Toffoli gate or an identity gate.} 
    \label{fig:circuit-construction}
\end{figure}

\subsection{Clock construction}

\subsubsection{Kitaev's construction}
Since Kitaev's original proof, all the circuit-to-Hamiltonian constructions heavily rely on a `clock register' to track when each gate of the verification circuit is applied. More concretely, let $\ket{\phi_0}$ be the input state vector of the verification circuit that is a product state of the state of the input qubits, auxiliary qubits, and the witness qubits. The so-called history state, which is crucially used to prove completeness condition for (Q)MA-hardness is the uniform superposition 
\begin{align}\label{eq:kitaev}
    \ket{\Phi} = \sum_{t = 0}^T U_t U_{t-1} \cdots U_1\ket{\phi_0} \otimes \ket{t}
\end{align}
of the `snapshots' of entire computation of the verification circuit, where $\ket{t}$ is the clock register that indicates the time (encoded in unary representation) corresponding to the `snapshot'. It is due to this clock register that such proof construction becomes geometrically non-local as we show by quoting Kitaev's propagation Hamiltonian construction in Eq.~\eqref{Eq-Kitaev-H_prop}. 

To ensure that low-energy states have the form Eq.~\eqref{eq:kitaev}, it is crucial to check at every time-step that the computation on the initial state is correctly propagated according to the local gates applied in the verification circuit.  In Kitaev's construction \cite{KitaevClassicalQuantumComputation}, we have the so-called {\em propagation terms} $H_{Prop}$ that ensure that the desired gate is applied when ``the clock ticks''. The propagation term from step $t-1$ to $t$ is done using the following Hamiltonian 
\begin{equation}\label{Eq-Kitaev-H_prop}
    H_{\rm Prop}^{\rm Kitaev} = \frac{1}{2}(\mathbbm{1}\otimes \ketbra{t}{t} + \mathbbm{1}\otimes \ketbra{t-1}{t-1} - U_t \otimes \ketbra{t}{t-1} - U_t^{\dagger}\otimes \ketbra{t-1}{t}).
\end{equation}
Even if $U_t$ is applied to any two neighbouring qubits in the \emph{computational} registers, it is not possible to make it geometrically local since they would be far from the \emph{clock} registers (denoted by $\ket{t}$ and $\ket{t-1}$). This construction therefore does not yield geometrically local Hamiltonians.

Recall that in Kitaev's construction, we have
\begin{equation}\label{eq:kitaev-full-hamilt}
    H = H_{\rm In} + H_{\rm Prop} + H_{\rm Out}
\end{equation}
and $H$ also contains two other terms, $H_{\rm In}$ and $H_{\rm Out}$. While $H_{\rm Out}$ can be made geometrically local by placing the last clock qubit and the output qubit side-by-side, $H_{\rm In}$ does not yield a geometrically local Hamiltonian. For an $n$-qubit input $\ket{x}$
\begin{equation}
    H_{\rm In} = \sum_{i=1}^n(\mathbbm{1} - \ketbra{x_i}{x_i})_i \otimes \ketbra{0}{0}
\end{equation} 
checks bit-by-bit whether the given input $\ket{x}$ is correct by projecting each input bit string onto its orthogonal subspace in the \textit{computational} register. Since the correctness of the input has to be ensured on the very first time step, the verifier first needs to project the input state to the starting time via $\ketbra{0}{0}_C$ projection on the clock register. We remind the reader that we maintain the convention of having the first register to be the \textit{computational} register and the second one to be the \textit{clock} register. We remark, in anticipation of the forthcoming discussion, that only $H_{\rm Prop}$ in Kitaev's Hamiltonian construction  (Eq.~\eqref{eq:kitaev-full-hamilt}) is non-stoquastic, in general. This is due to $U_t$ being any arbitrary $2$-qubit unitary. Much of our effort in the following sections shall be spent in making $H_{\rm Prop}$ stoquastic for a special choice of unitaries in a way that removes the need for separate clock registers, thus making them \textit{geometrically} $2$-local.

\subsubsection{Clock propagation}\label{subsection:clock-prop}
In order to achieve a geometrically local construction, the idea presented in 
Ref.\ \cite{aharonov2008adiabatic} is to consider computational and clock information altogether, by ``embedding'' the clock into the particles on the grid (or a line) such that the `shape' (as we shall describe below) of the system corresponds to a unique point in time in the verification circuit. 

The immediate upshot of encoding more information (about the clock, in this case) into a particle is that now we have to work with higher dimensional particles --- qudits. We shall represent these basis states of a qudit by \emph{shapes}, as it will be easier to show how these shapes change as the circuit verification proceeds to help us identify where we are in the verification circuit. To that end, we introduce different states that a particle could be in. Firstly, we represent the computational subspace, on which the unitaries act, using the 
state vectors $\ket{\stateba}, \ket{\statebb}, \ket{\stateca},\ket{\statecb}, \ket{\statebr}, \ket{\statebl}, \ket{\statehr}$, and $\ket{\statehl}$. Here, each pair represents a one-qubit computational basis. For example, one can think of $\ket{\statehr}$ representing $\ket{0}$ and $\ket{\statehl}$ representing $\ket{1}$, and similarly for the other two pairs. An arbitrary superposition of a state vector is represented by $\ket{\statehrl}$ in the subspace spanned by $\{\ket{\statehr}, \ket{\statehl}\}$ within the computational subspace. We need two more state vectors to be able to track time in our state of computation: $\ket{\statea}$ and $\ket{\stated}$, which we term as \textit{unborn} and \textit{dead} states, respectively. We remark that we use similar shapes and names as the ones used in Ref.\ \cite{aharonov2008adiabatic} in order to make the exposition smooth for the readers who are already familiar with prior work on geometric constructions \cite{nagaj-1d, aharonov2008adiabatic}.

The utility of working with the aforementioned shapes is that it enables us to keep track of the clock by mapping the gate application in the verification circuit to the shape of the grid and vice-versa. Since gate application in the original circuit always involves applying gates to some input qubits, let us first see how gate application would look like when we work with shapes. As explained earlier, one can think of a single a pair of shapes with complementary arrow directions as a pair of basis states for a qubit Hilbert space: $\{\ket{\stateba}, \ket{\statebb}\}, \{\ket{\stateca}, \ket{\statecb}\}, \{\ket{\statebr}, \ket{\statebl}\},
\{\ket{\statehr}, \ket{\statehl}\}
$.

In a specific a single layer of computation as described in \Cref{sec:layers}, 
the horizontal and vertical lines are used to denote two qudits that are bundled together (which we will shortly describe).
Moreover, the single-lines denote that no gates in that phase have already acted on the qubit, while the double denote that at least one gate has already acted on the relevant qubit.

While we almost exclusively work with \emph{three}-qubit Toffolis, let us, for pedagogical reasons, start with only two-qubit unitaries. Let $U_{i, i+1}$ be a two-qubit unitary acting on qubits $i, i+1$ (in the original verification circuit). Pictorially, this is represented as
\begin{equation} \label{Eq-unitary-shapes-map-2qubits}
U_{i, i+1} \left( \stackket{\statecboxed}{\statebboxed}_{i, i+1} \right) = \stackket{\statecboxed}{\statecboxed}_{i, i+1},
\end{equation}
where we are assuming that a unitary $U_{i-1, i}$ already acted on qubits $i-1, i$ due to which the state of the top most particle is already double-lined.

This is essentially the same idea as it is used in the first geometric construction of  
Ref.~\cite{aharonov2008adiabatic} and most of the follow-up work since then \cite{nagaj-1d, aharonov2008adiabatic}. Note that one can define a similar rule as in Eq.~\eqref{Eq-unitary-shapes-map-2qubits} for three-qubit unitaries acting on three particles. Doing so, one can realize arbitrary three-qubit Toffolis (universal for classical computation) using \textit{three}-local Hamiltonians by using the construction in Ref.\ \cite{aharonov2008adiabatic}. Furthermore, this construction would also preserve stoquasticity because the `propagation' Hamiltonian as discussed in Eq.~\eqref{Eq-Kitaev-H_prop} shall become stoquastic because $-U$ will always have non-positive off-diagonals for all Toffolis. 
All the remaining terms (in the construction of the local Hamiltonian problem), such as the Hamiltonian checking the correctness of input projects onto the space orthogonal to correct input space, the `output Hamiltonian' projects onto the non-accepting subspace of the verification circuit and the `penalty Hamiltonian' projects onto the space of all the forbidden configurations. Since each of these Hamiltonians is a projector in the computational basis, they are trivially stoquastic. Therefore, the construction 
that has been introduced in Ref.\ \cite{aharonov2008adiabatic} directly gives us a recipe for proving MA-hardness of \textit{three}-local stoquastic Hamiltonian problem on a two-dimensional grid. 

Since most physical Hamiltonians of interest are nearest-neighbour Hamiltonians, it is interesting to ascertain whether the hardness shall also hold for nearest-neighbour (geometrically 2-local) Hamiltonians on a two-dimensional grid. We show that hardness does hold even for nearest-neighbour Hamiltonians by encoding the action of a three-qubit unitary on three particles into only two (nearest-neighbour) particles. We extend the mapping rule in Eq.~\eqref{Eq-unitary-shapes-map-2qubits} to three-qubit unitaries by encoding two logical qubits in a single qudit with higher dimension. To do this, we put together two different basis states from the earlier construction into one (higher-dimensional) particle. For example, the joint state of two qubits, say, $\ket{\stateb}$ (arbitrary superposition over $\stateba$ and $\statebb$) and $\ket{\statebh}$ (arbitrary superposition over $\ket{\statebr}$ and $\ket{\statebl}$)  can be represented by a new higher dimensional particle to be $\dstatebb$. Now, the unitary action on this particle encoding two logical qubits can be tracked by double-lining either one or both of $\stateb,\statebh$ depending on whether, within the same particle, the unitary has acted on one logical qubit (denoted as $\dstatecb$) or both of them (denoted as $\dstatecc$).
Example of such three-qubit unitary action on a high-dimensional qudit is shown in Eq.~\eqref{eq-2D-unitary-action-example}, which allows for the application of arbitrary \textit{three-qubit} unitaries on \textit{nearest-neighbour} particles. 

We demonstrate that the shapes on the grid track computational steps. The process begins with a single qubit unitary, followed by a 2-qubit unitary (Fig.~\ref{fig:lattice_1-3_time-steps}), continuing with successive 3-qubit Toffoli gates (Fig.~\ref{figure:clock-prop}). After the first two steps, Toffoli gates involve two previously engaged qubits and one new, ensuring single arrow transitions within each particle at each step. This ensures that the change is only in adjacent \textit{particles}' states. Upon completing the downward phase, the column becomes $\dstatecc$. The upward phase starts with identity gates resetting the unitary-acted states, transitioning the right column's unborn state to $\dstatebb$ and marking the previous $\dstatecc$ as dead $\dstated$. This allows two-local checks against ``illegal'' states due to the two-site grid change. The Hilbert space of any particle can be decomposed as shown in Eq.~\eqref{eq-hilbert-space-shape-decomposition}. Each shape on the right hand side in Eq.~\eqref{eq-hilbert-space-shape-decomposition} is the shape of the respective Hilbert space states associated to some point in the computation, 
\begin{equation}\label{eq-hilbert-space-shape-decomposition}
\mathcal{H}_{\rm Particle} = \stateda \oplus \dstated \oplus \dstatebb \oplus \dstatecb \oplus \dstatecc, 
\end{equation}
where the first two Hilbert spaces are one-dimensional each and the last three are 4 dimensional each --- since they track two  qubits. Consequently, the total Hilbert space dimension is given by $1+1+4+4+4 = 14$. 

\begin{figure}[ht]
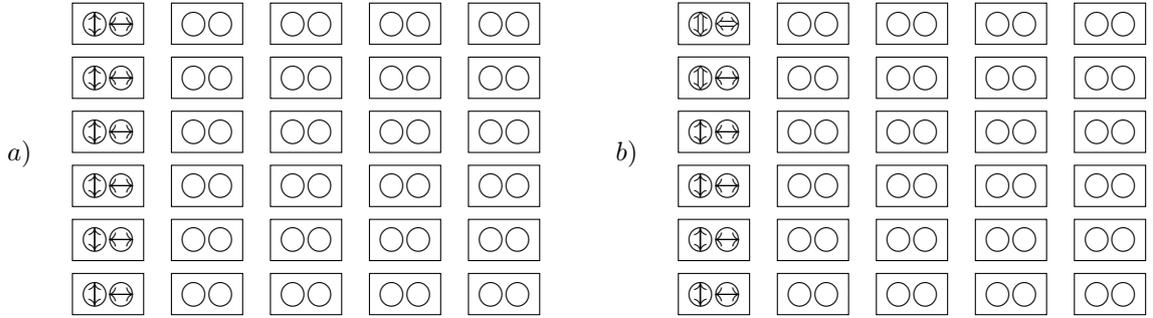

    \center{{
\begin{equation*}
\begin{array}{cccc}
a) &
\begin{array}{cccccc}
\dstatebb & \stateda & \stateda & \stateda & \stateda  \\
\dstatebb & \stateda & \stateda & \stateda & \stateda  \\
\dstatebb & \stateda & \stateda & \stateda & \stateda  \\
\dstatebb & \stateda & \stateda & \stateda & \stateda  \\
\dstatebb & \stateda & \stateda & \stateda & \stateda   \\
\dstatebb & \stateda & \stateda & \stateda & \stateda  
\end{array}
&~~~~b)&
\begin{array}{cccccc}
\dstatecc & \stateda & \stateda & \stateda & \stateda \\
\dstatecb & \stateda & \stateda & \stateda & \stateda \\
\dstatebb & \stateda & \stateda & \stateda & \stateda \\
\dstatebb & \stateda & \stateda & \stateda & \stateda  \\
\dstatebb & \stateda & \stateda & \stateda & \stateda  \\
\dstatebb & \stateda & \stateda & \stateda & \stateda  
\end{array}
\end{array}\nonumber
\end{equation*}
    \caption{Both computational states at the start of the circuit prior to being acted upon by any unitary in (a). Subfigure (b) represents the state of the lattice at the third time step when the first 1-, 2-, and 3-qubit unitaries have acted (as shown in Fig.~\ref{fig:circuit-construction}) --- each changing the state of a single-lined particle to a double-lined particle, one at a time.}
    \label{fig:lattice_1-3_time-steps}
}}
\end{figure}

\begin{figure}[ht]
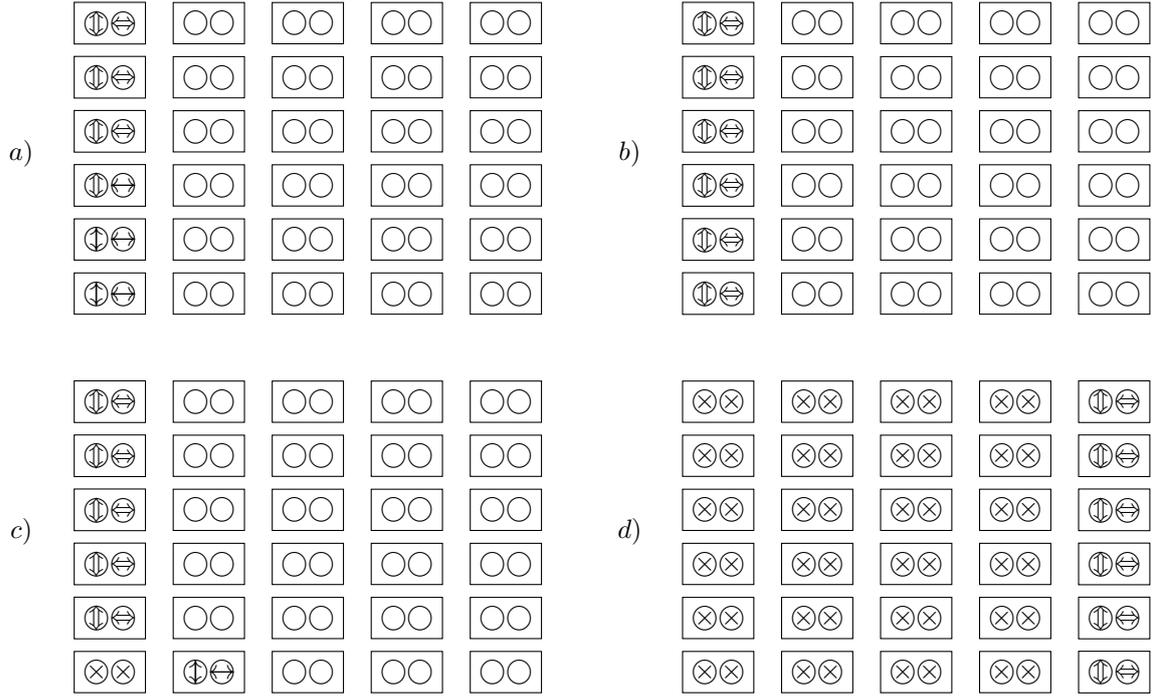

\center{{
\begin{equation*}
\begin{array}{cccc}
a) &
\begin{array}{cccccc}
\dstatecc & \stateda & \stateda & \stateda & \stateda \\
\dstatecc & \stateda & \stateda & \stateda & \stateda \\
\dstatecc & \stateda & \stateda & \stateda & \stateda \\
\dstatecb & \stateda & \stateda & \stateda & \stateda \\
\dstatebb & \stateda & \stateda & \stateda & \stateda \\
\dstatebb & \stateda & \stateda & \stateda & \stateda  
\end{array}&~~~~b) &
\begin{array}{ccccc}
\dstatecc & \stateda & \stateda & \stateda & \stateda \\
\dstatecc & \stateda & \stateda & \stateda & \stateda \\
\dstatecc & \stateda & \stateda & \stateda & \stateda \\
\dstatecc & \stateda & \stateda & \stateda & \stateda \\
\dstatecc & \stateda & \stateda & \stateda & \stateda \\
\dstatecc & \stateda & \stateda & \stateda & \stateda
\end{array}\\\nonumber
&&&\\\nonumber
c) &
\begin{array}{ccccc}
\dstatecc & \stateda & \stateda & \stateda & \stateda \\
\dstatecc & \stateda & \stateda & \stateda & \stateda \\
\dstatecc & \stateda & \stateda & \stateda & \stateda \\
\dstatecc & \stateda & \stateda & \stateda & \stateda \\
\dstatecc & \stateda & \stateda & \stateda & \stateda \\
\dstated & \dstatebb & \stateda & \stateda & \stateda
\end{array}&~~~~d)&
\begin{array}{ccccc}
 \dstated & \dstated & \dstated & \dstated & \dstatecc \\
 \dstated & \dstated & \dstated & \dstated & \dstatecc \\
 \dstated & \dstated & \dstated & \dstated & \dstatecc \\
 \dstated & \dstated & \dstated & \dstated & \dstatecc \\
 \dstated & \dstated & \dstated & \dstated & \dstatecc \\
 \dstated & \dstated & \dstated & \dstated & \dstatecc
\end{array}
\end{array}\nonumber
\end{equation*}
}}
\caption{
State of the grid corresponding to various stages during the MA circuit computation. Subfigure (a) represents a step in the middle of the downward phase, (b) represents the last step in the downward phase, while (c) represents the immediate next step following (b): the first step in the upward phase. Finally, (d) represents the final time step.}
\label{figure:clock-prop}
\end{figure}

\subsection{The penalty Hamiltonian}
At each time step, we have to ensure that the shape of our grid evolves from one `legal' shape to another. For example, single lined particles should always by vertically above the double-lined ones, dead states should always be on the left of the unborn states, etc. More concretely, on a grid with $n'$ rows and $2R + 1$ columns (where $n'$ and $R$ will correspond to the number of particles and the number of layers respectively in the MA verification circuit), a state is legal if there exist a column index $r$ and row index $k$ that satisfy the following rules.
\begin{itemize}
    \item Downward phase: The first (leftmost) $r$ columns are in the dead phase. The top $k$ particles of the $r+1$st column are in the $\dstatecc$ phase and the bottom $n'/2 - k$ particles are in the $\dstatebb$ phase. The last $2R+1-r$ columns are in the unborn phase $\stateda$.  
    \item Upward phase: The first (leftmost) $r$ columns are in the dead phase. The $r+1$st column has its top $n'/2 -k$ particles in the $\dstatecc$ state and the remaining $k$ particles in the $\dstated$ phase. The $r+2$nd column has its top $n'/2 -k$ particles in the $\stateda$ phase and its remaining particles in the $\dstatebb$ phase. All the remaining rightmost columns are in the $\stateda$ phase.
 \end{itemize}

Shapes that do not have this property are said to be illegal. We notice that all illegal shapes can be detected by local patterns that we describe in  Tables \ref{tab-intra-forbidden} and \ref{tab-inter-forbidden}. In order to penalize these illegal shapes, we define {\em penalty} terms from these patterns, which ensure that whenever the grid has any illegal configuration, the energy penalty is at least 1 for at least one of the penalty terms.

Hence, the penalty Hamiltonian is the sum over projectors onto each of states corresponding to the illegal patterns shown in Tables \ref{tab-intra-forbidden} and \ref{tab-inter-forbidden}. Thus,
\begin{equation}\label{eq-hamilt-clock}
    H_{\rm Penalty} = \sum_{i \in {\rm Illegal}} \Pi_i,
\end{equation}
where $\Pi_i$'s are the projectors onto the illegal clock states. We introduce some illegal states in Table~\ref{tab-intra-forbidden} here and defer the full Table of illegal states to Appendix~\ref{sec-illegal-states}, where we list the forbidden configurations between the neighbouring particles. Note that both Tables~\ref{tab-intra-forbidden} and \ref{tab-inter-forbidden} involve configurations that do not involve particles beyond nearest neighbours. Thus, they can always be penalized using a two-local projector onto the said configuration, similar to Table 1 in Ref.~\cite{aharonov2008adiabatic}. 
However, since in our case, we pair the original qubits together to be able to deal with Toffoli gates, we need to take care of higher dimensional particles, there are additional particle configurations. Ref.~\cite{aharonov2008adiabatic} only works with particles that encode one logical qubit in a six-dimensional particle (including the non-computational flags). Since we encode two logical qubits (and other non-computational states) using 14-dimensional particles, we encounter many more configurations that could be illegal and need to penalized. For example, states in our Table~\ref{tab-intra-forbidden} cannot feature, by construction, in Ref.~\cite{aharonov2008adiabatic} since they only ever need to penalize \textit{inter}-particle illegal configurations, we encounter \textit{intra}-particle states as shown in Table~\ref{tab-intra-forbidden}.

Since we apply two qubit (and not single qubit) identities in the upward phase, we get that unborn and dead states always appear in pairs (Rules 1 and 2 in Table~\ref{tab-intra-forbidden}). Otherwise, we would need more dimensions to handle the cases with one subparticle \textit{dead} and one \textit{alive}. For example, $\stateda$. Having such configurations as legal would have increased the particle dimension of our qudit, which we avoid by simply applying 2-qubit identities (instead of 1 qubit identities) in the reset (upwards) phase via the $H_{t}^{\rm Up}$ Hamiltonian in Eq.~\eqref{Eq-H_prop-Up}. Had we applied 1 qubit unitaries, it would have incurred extra particle dimensions we would have to keep track of. We remark that despite this, one can still apply 1-qubit unitaries, as is the case in Ref.~\cite{aharonov2008adiabatic}. The upshot, however, shall be twofold: (i) more number of steps needed to reset in the upward phase, and (ii) increased qudit dimension.

\begin{table}[H]
\center{
\begin{tabular}{|l|l|}
\hline
Forbidden  & Guarantees that \\ \hline \hline
   $\fstateea, \fstateae$ & Unborn state in a single particle must exist in pairs. \\ \hline
   $\fstatede, \fstateed$ & Dead states in a particle must exist in pairs. \\ 
   \hline
   $\fdstatecb$ & Unitary acts first on the left computational ``substate''. \\ 
   \hline
\end{tabular}
}
\caption{Table of illegal states within the same particle that can be penalized using only two-local Hamiltonian terms. Here,  $\ket{\stateelse}$ represents all state vectors except the neighbouring state --- within the same particle. Since these states are within a single particle, they only need a one-local check.}
\label{tab-intra-forbidden}
\end{table}

\subsection{The propagation Hamiltonian}\label{sec:prophamilt-2D}

In the standard local Hamiltonian problem, the Hamiltonian is not stoquastic because the propagation terms (See Eq.~\eqref{Eq-Kitaev-H_prop}) are not, in general, stoquastic. Choosing $U$ to be a Toffoli gate (universal for classical computation), we get a stoquastic Hamiltonian construction, since $U \otimes (- \ket{t}\bra{t+1})$ is always stoquastic, which is the term that appears in the propagation term. We shall now illustrate how to make this stoquastic construction geometrically two-local as well by explicitly constructing the propagation Hamiltonians for the downward (action of unitaries) and upward (reset by identity gates) phases, as shown in Fig.~\ref{fig:circuit-construction}.

For the downward case, steps during the middle of a computation (that is, computational steps involving non-initial and non-final qubits), and the edge cases can be worked very similarly.~\footnote{For an exposition on edge cases for QMA-hardness, see Ref.\ \cite{aharonov2008adiabatic}.} The propagation Hamiltonian is constructed as follows: Let our original MA-verification circuit be composed of $n'$ particles with $R$ layers/rounds. Then the modified circuit, as shown in Fig.\ \ref{fig:circuit-construction}, 
will have $n'R + n'R/2$ gates. And since we encode two qubits in a single particle, the two-dimensional grid will have exactly $n'/2$ sites on the vertical direction, while the horizontal sites (corresponding to the depth of the original circuit) will be $2R+1$, where the very first layer is simply used to check for correct inputs and auxiliary qubits. Hence, we shall work with a $\frac{n'}{2} \times \left(2R+1\right)$ grid, a two-dimensional lattice.

Let us assume that the Hamiltonian acts on column $r$ between rows $k-1$ and $k$. Then the action of the Hamiltonian is, like in the algebraic case, to either evolve the state one step forward in time, one step backward or leave it in the current state. The geometric version at time step $t$ is shown below, where the first summand takes care of forward and backward propagation and the remaining summands correspond to applying identity operation of time steps $t$ and $t-1$,
\begin{align} \label{Eq-H_prop-Down}
H_{t}^{\rm Down} := \left(\begin{array}{cc} 0& -U_t\\ -U_t^{\dagger}&0
\end{array}\right) &+
\left(
\stackketbra{\statecahr}{\statebabhr} +
\cdots + 
\stackketbra{\statecbhl}{\statebbhl}
\right)\ontop{_{k-1,r}}{_{k,r}} \\
\nonumber
&+
\left(
\stackketbra{\statecahr}{\statecabhr} + \cdots +
\stackketbra{\statecbhl}{\statecbhbl}
\right)\ontop{_{k-1,r}}{_{k,r}} \\
\nonumber
&+
\left(
\stackketbra{\statecahr}{\statebabhr} +
\cdots + 
\stackketbra{\statecbhl}{\statebbhl}
\right)\ontop{_{k,r}}{_{k+1,r}} \\
\nonumber
&+
\left(
\stackketbra{\statecahr}{\statecabhr} + \cdots +
\stackketbra{\statecbhl}{\statecbhbl}
\right)\ontop{_{k,r}}{_{k+1,r}}.
\nonumber
\end{align}
There are a total of 16 summands (four choices in each of the two rows, indexed by $k\pm1$) in each of the last four terms in the Hamiltonian above. Note that all summands except the first are projectors, hence do not affect stoquasticity. The $U_\ell$ in the first term is a Toffoli gate such that $-U_\ell$ is always stoquastic. Similarly, one can write down the Hamiltonian for the upwards phase as 
\begin{align} \label{Eq-H_prop-Up}
H_{t}^{\rm Up} &:= \Bigg( \Bigg. \stackketbra{\statecahr}{\dstated} \hskip -5pt \ontop{_{i,r}}{_{i+1,r}} + \stackketbra{\stateda}{\statebabhr} \hskip -5pt \ontop{_{i-1,r+1}}{_{i,r+1}} + \cdots \\ 
\nonumber
&+ \stackketbra{\statecbhl}{\dstated} \hskip -5pt \ontop{_{i,r}}{_{i+1,r}} + \stackketbra{\stateda}{\statebbhl}\hskip -5pt \ontop{_{i-1,r+1}}{_{i,r+1}}\Bigg. \Bigg)   \\
\nonumber
&-
\Bigg( \Bigg. 
\ketbra{\dstated \statebabhr}{\statecahr \stateda} + \ketbra{\statecahr \stateda}{\dstated \statebabhr} + \cdots \\
\nonumber
&+ \ketbra{\dstated \statebbhl}{\statecbhl \stateda} + \ketbra{\statecbhl \stateda}{\dstated \statebbhl}
\Bigg. \Bigg)_{\rm(i,r)(i,r+1)} .
\nonumber
\end{align}
Here, the terms in the first parenthesis correspond to simply applying identity on time steps $t-1$ and $t$. So one simply enlists projectors onto all possible combinations of shapes in the two consecutive time steps. In the upward phase, unlike in the downward phase, the consecutive time steps involve neighbouring columns as well as neighbouring rows. The summands in the second parenthesis involve applying $-U_t = -\mathbbm{1}$. Again, this comes about by mapping all the basis shapes that could occur at time step $t$ to $t-1$ and vice-versa. 

\subsection{Proof of Theorem~\ref{theorem:2d-ma-hard}}
We are now ready to spell out the proof of Theorem~\ref{theorem:2d-ma-hard}, which we shall split into the following two sections: Completeness (Section \ref{subsection:completeness-proof}) and soundness (Section \ref{subsection:soundness-proof}). First, let us define the full Hamiltonian 
\begin{equation} \label{Eq-Full-Hamilt-2}
    H = H_{\rm Init} + H_{\rm Prop}/2 + H_{\rm Final} + H_{\rm Penalty},
\end{equation}
where $H_{\rm Prop}$ is the propagation Hamiltonian denoted as $H_{t}^{\rm Down}$ and $H_{t}^{\rm Up}$ in Eqs.~\eqref{Eq-H_prop-Down} and \eqref{Eq-H_prop-Up} respectively. And $H_{\rm Penalty}$ is that defined in Eq.~\eqref{eq-hamilt-clock}. We now define the remaining two Hamiltonian terms: $H_{\rm Init}$ and $H_{\rm Final}$. Firstly, assume that the $n$-particle input vector $\ket{x}$, $w$-particle witness $\ket{\psi}$, $r$-particle auxiliary systems in an equal superposition of $\ketbra{\statebabhr}{\statebabhr}$ and $\ketbra{\statebbhl}{\statebbhl}$ (to simulate classical randomness needed for the class MA$_{\text{q}}$), and an additional $s$ auxiliary systems are arranged top-to-bottom in the first column of our two-dimensional grid on the first $n, w, r$ and $s$ (high-dimensional) particles, respectively. We can always do this because the input and the witness are given to us as tensor products on respective Hilbert spaces. Moreover, we can always append auxiliary Hilbert spaces of convenient dimensions. One can then write $ H_{\rm Init}$ as \begin{equation}
        H_{\rm Init} := \sum_i \left(\mathbbm{1} - \ketbra{\statebabhr}{\statebabhr}\right)_{i,1} + 
        \sum_{i'}\left(\mathbbm{1} - \ketbra{\text{Coin}}{\text{Coin}}\right)_{i',1}
        + \sum_{j}\left(\ketbra{\neg x_j}{\neg x_j}\right)_{j,1},
    \end{equation}
    where the `Coin' qubits
     \begin{equation}
    \ket{\text{Coin}} := \frac{1}{\sqrt{2}}\ketbra{\statebabhr}{\statebabhr} + \frac{1}{\sqrt{2}}\ketbra{\statebbhl}{\statebbhl}
    \end{equation}
    introduce classical randomness for MA-verification. The indices $i$ denote the sum over the auxiliaries, $i'$ is the sum over the `Coin' qudits, and $j$ is the sum over the inputs. In the geometric construction, we work with the modified circuit as shown in Fig.~\ref{fig:circuit-construction}. We expect that our $s$ auxiliary particles in $\ket{\phi_0}$ are all initialized as $\ketbra{\statebabhr}{\statebabhr}$. Finally, $H_{\rm Final}$ is defined as  
    \begin{equation}
        H_{\rm Final}\coloneqq  \left(\mathbbm{1} - \ketbra{\statecbhl}{\statecbhl}\right)_{i = 1, j = 2R+1},
    \end{equation}
 $i = 1$ and $j= 2R+1$ corresponds precisely to the top right column since that would correspond to $l= 2R +1$, which is indeed the very last identity gate in the original construction. The additional $+1$ appears because in the first column on our grid, we shall only apply identities to verify of the input and auxiliary qubits are in the desired state. As for stoquasticity, we already designed a stoquastic construction for $H_{\rm Prop}$ in Section~\ref{sec:prophamilt-2D}. And since $H_{\rm Final}, H_{\rm Init}$, and $H_{\rm Penalty}$ are projectors, they are already stoquastic. Thus the full Hamiltonian (Eq.~\eqref{Eq-Full-Hamilt-2}) is stoquastic. We are now ready to prove completeness and soundness conditions for proof of Theorem~\ref{theorem:2d-ma-hard}. For both the completeness and soundness conditions, we work with the class MA$_{\text{q}}$ due to its equivalence with the class MA since that makes it easier to use existing tools in Hamiltonian complexity.

\subsubsection{Completeness}\label{subsection:completeness-proof}

Here we will show that if $x$
is a `yes' instance of MA$_{\text{q}}$, then this case can be reduced to the `yes' instance of the local Hamiltonian problem. 

\begin{lemma}[Completeness]\label{lemma-completeness-unamplified}
    Let $x$ be a `yes' instance of  MA$_{\text{q}}$ that is accepted by a verification circuit $U$ of depth $T = {\rm poly}(n)$ with probability at least $1 - \varepsilon$ for some $n$-qubit witness $\ket{\psi}$. Then the geometrically local Hamiltonian $H = H_{\rm Init} + H_{\rm Prop}/2 + H_{\rm Final} + H_{\rm Penalty}$ constructed in \Cref{Eq-Full-Hamilt-2} admits $\lambda_{\min}(H) \leq O(\varepsilon/T)$. 
\end{lemma}
\begin{proof}
For any `yes' instance, $x \in L_{\rm Yes}$, of MA$_{\text{q}}$,  we know that $\exists \ket{\psi}$ such that $U(\ket{x}\ket{\psi})$ accepts with probability at least
$1-\varepsilon$ for a quantum verification circuit $U = U_T\cdot U_{T-1}\cdots U_2\cdot U_1$. We claim that our geometric construction of the full Hamiltonian $H$ (Eq.~\eqref{Eq-Full-Hamilt-2}) admits the desired $\lambda_{\min}(H) \leq O(\varepsilon/T)$ bound for the state which is the superposition of the so-called history states of the quantum verification circuit. We denote this history state vector by 
\begin{equation}
\ket{\Phi} \coloneqq \frac{1}{\sqrt{T+1}}\sum_{t=0}^T \ket{\phi_t}, 
\end{equation}
for some input state vector $\ket{\phi_0}$ (a tensor product state of the input, the witness, and the auxiliary systems on the grid), for not necessarily the correct input state. Each $\ket{\phi_t} = \mathsf{Enc}(U_t U_{t-1} \cdots U_1\ket{\phi_0})$, where $\mathsf{Enc}(.)$ is the 2D encoding (using the geometric construction described in \Cref{subsection:clock-prop}). We claim that $\ket{\Phi}$ is the (low) energy state of the full Hamiltonian $H$ with energy $\leq O(\varepsilon/T)$. First note that since $H_{\rm Init}$ projects all the input bits to their orthogonal subspace, we have that
\begin{equation}
    H_{\rm Init}\ket{\Phi} = 0 .
\end{equation}
Next, the penalty (sometimes also called clock) Hamiltonian will penalize the ``illegal'' clock states. This shall ensure that we always are transitioning to the correct next state. To penalize the wrong time evolution, one can simply sum over all the projectors onto the wrong evolution configurations. This would ensure that all the correct configurations will get a zero penalty and any wrong configuration shall suffer a non-zero penalty. Using Tables \ref{tab-intra-forbidden} and \ref{tab-inter-forbidden}, we see that all the illegal configurations differ in at most two neighbouring sites. Hence, we can construct two-local projectors projecting on each of the forbidden configurations
\begin{equation}
    H_{\rm Penalty} = \sum_{i \in {\rm Illegal}} \Pi_i,
\end{equation}
since any legal state (and superpositions thereof) are in the nullspace 
of $H_{\rm Penalty}$,
\begin{equation}
    H_{\rm Penalty} \ket{\Phi} = 0.
\end{equation}
Finally note that $ \ket{\Phi}$ is also the zero eigenstate of $H_{\rm Prop}$, like the algebraically-local case \cite{aharonov2008adiabatic}. 
Now, observe that $\bra{\Phi}H_{\rm Final} \ket{\Phi}$ only has a non-zero contribution on $\ket{\phi_{T+1}}$ from the entire history state vector $\ket{\Phi}$ 
as
\begin{equation}
\bra{\Phi}H_{\rm Final} \ket{\Phi} = \frac{1}{T+1} \bra{\phi_{T+1}}H_{\rm Final}\ket{\phi_{T+1}} \leq \frac{1}{T+1}\norm{H_{\rm Final}\ket{\phi_{t+1}}}_{2}^2  \leq \frac{\varepsilon}{T+1},
\end{equation}
where the rejection probability (which is the norm squared of $H_{\rm Final}\ket{\Phi}$) is at most $\varepsilon$.
Hence, if $x \in L_{\rm Yes}$, $\bra{\Phi}H_{\rm Final} \ket{\Phi} \leq \frac{\varepsilon}{T}$. Finally, one can note that $H_{\rm Init}, H_{\rm Final}, H_{\rm Penalty}$ are stoquastic because they are all projectors. Since we already argued in Section \ref{sec:prophamilt-2D} that $H_{\rm Prop}$ is stoquastic, this ensures that our full Hamiltonian in Eq.~\eqref{Eq-Full-Hamilt-2} is also stoquastic. This concludes the completeness proof.
\end{proof}

\begin{corollary}
Let $U$ be an MA$_{\text{q}}$ verification circuit with depth $T= {\rm poly}(n)$, and let $x$ be a `yes' instance that is accepted by $U$ with probability $1$ for some $n$-qubit witness $\ket{\psi}$. Then the geometrically local Hamiltonian $H$ constructed in \Cref{Eq-Full-Hamilt-2} has $\lambda_{\min}(H) = 0$.
\end{corollary}

\subsubsection{Soundness} \label{subsection:soundness-proof}
The proof of soundness follows essentially from algebraically local case of QMA-hardness of the general LH-MIN problem. We still state the proof for completeness. Our proof closely resembles existing proofs presented in
Refs.\ \cite{aharonov2002quantum,KitaevClassicalQuantumComputation,Aharonov_2009}.

\begin{lemma}[MA$_{\text{q}}$ soundness]\label{lemma-soundness-qma}
    If $x$ is a `no' instance of MA$_{\text{q}}$, then $\lambda_{\min}(H) = \Omega(1/T^3)$, where $T = {\rm poly}(n)$ and $n$ is the number of qubits of the input $\ket{x}$.
\end{lemma}

 \begin{proof}
 Let us decompose $H = A_1 + A_2$, where $A_1 = H_{\rm Init} + H_{\rm Final}$ and $A_2 = H_{\rm Prop}$. Since $H_{\rm Init}$ and $H_{\rm Final}$ act on different time leaves (geometrically, on different parts of the two-dimensional grid), they commute and hence share an eigenbasis. Since all the summands in $H_{\rm Init}$ are projectors, the minimum eigenvalue of $A_1$ is 1. The difficulty in lower bounding $\lambda_{\min}(H)$ is that $A_1$ and $A_2$ do not commute. To that end, we shall use a geometric lemma \cite{KitaevClassicalQuantumComputation,aharonov2002quantum} 
 that lower bounds the minimum eigenvalue of the sum of two Hermitian PSD matrices.

 \begin{lemma}[Lemma 1 in Ref.\  \cite{aharonov2002quantum}]\label{lemma-geometric-kitaev}
     Let $H_1$ and $H_2$ be two Hermitian positive semi-definite matrices, and let $N_1$ and $N_2$ be the eigenspaces of the eigenvalue 0, respectively. If the angle between $N_1$ and $N_2$ is some $\theta > 0$, and the second eigenvalue of both $H_1$ and $H_2$ is $\geq \lambda$, then the minimal eigenvalue of $H_1 + H_2 \geq \lambda \sin^2(\theta/2)$.
 \end{lemma}
 
We know from the results of Ref.\ \cite{KitaevClassicalQuantumComputation} that the second smallest (or the smallest non-zero) eigenvalue of $H_{\rm Prop}$ is $\Omega(1/T^2)$, where $T$ corresponds to the total steps in the verification circuit of MA$_{\text{q}}$. On the other hand, as argued above, the smallest non-zero eigenvalue of $H_{\rm Init} + H_{\rm Final}$ is 1. Thus $\lambda$ in Lemma~\ref{lemma-geometric-kitaev} is $\Omega(1/T^2)$. What remains is to find a lower bound on $\sin^2(\theta/2)$. Let $\Pi_{\rm Init}$ and $\Pi_{\rm Final}$ be the projectors onto the nullspaces of $H_{\rm Init}$ and $H_{\rm Final}$ respectively. Then, one can write the total projection onto the nullspace of $A_1$ as $\Pi_{\rm Init}\Pi_{\rm Final}$ since each of the projectors acts on different subspaces. As argued in Ref.\ \cite{aharonov2002quantum}, $\cos^2(\theta)$ is the maximum $2$-norm squared of the projection of the history state, $\Phi$, onto the nullspace of $A_1$. We shall see that it suffices to upper bound $\norm{\Pi_{\rm Init}\Pi_{\rm Final}\ket{\Phi}}^2$ and using $\sin^2(\theta) + \cos^2(\theta) = 1$ to get the desired lower bound on $\sin^2(\theta)$. Now, let us decompose the state $\ket{\phi_0}$ as superposition over valid and an invalid input
\begin{equation}
    \ket{\phi_0} = \alpha \ket{\phi_0^v} + \beta \ket{\phi_0^i},
\end{equation}
where $\abs{\alpha}^2$ is the probability that the prover gives a valid input state with correctly initialized auxiliaries in state $\ket{0}^{\otimes s}$. We label all such valid state vectors by $\ket{\phi_0^v}$. Any state where any one of the input of auxiliary systems are incorrectly initialized are henceforth represented by $\ket{\phi_0^i}$ in the orthogonal subspace. By linearity of unitary transformations, it holds $\ket{\Phi} = \alpha \ket{\Phi^v} + \beta \ket{\Phi^i}$, where $\ket{\Phi^v}$ (resp. $\ket{\Phi^i}$) are the valid (resp. invalid) history states corresponding to the input state vectors $\ket{\phi_0^v}$ and $\ket{\phi_0^i}$,
\begin{eqnarray}
    \norm{\Pi_{\rm Init}\Pi_{\rm Final}\ket{\Phi}}^2 &=& \bra{\Phi} \Pi_{\rm Init}\Pi_{\rm Final} \ket{\Phi}\\ 
    \nonumber
    &=&
    \abs{\alpha}^2 \norm{\Pi_{\rm Init}\Pi_{\rm Final}\ket{\Phi^v}}^2 + \abs{\beta}^2\ \norm{\Pi_{\rm Init}\Pi_{\rm Final}\ket{\Phi^i}}^2 + z + z^{*},
\end{eqnarray}
where $z := \alpha^*\beta\bra{\Phi^v}\Pi_{\rm Init}\Pi_{\rm Final}\ket{\Phi^i}$.

We shall upper bound each of the summands separately. For the first summand, note that $\Pi_{\rm Init}$ shall not contribute anything for a valid input state, the intermediate timeleaves shall contribute a factor of $T$ and $\Pi_{\rm Final}$ shall contribute a factor only on the top right corner of the grid (in the final column)
\begin{equation}
     \norm{\Pi_{\rm Init}\Pi_{\rm Final}\ket{\Phi^v}}^2 \leq \frac{1}{T+1} (T + \bra{\phi_{T+1}^v}\Pi_{\rm Final}\ket{\phi_{T+1}^v}) \leq \frac{T + \varepsilon}{T+1},
\end{equation}
where $\varepsilon$ is the small acceptance probability whenever $x$ is not in the language. For the second summand, note that 
\begin{equation}
   \norm{\Pi_{\rm Init}\Pi_{\rm Final}\ket{\Phi^i}}^2 =  \norm{\Pi_{\rm Final}\Pi_{\rm Init}\ket{\Phi^i}}^2 \leq \norm{\Pi_{\rm Init}\ket{\Phi^i}}^2,
\end{equation}
where the last inequality follows because (by Cauchy-Schwarz inequality) $\norm{P \ket{x}} \leq \norm{\ket{x}}$ for a projection $P$. By definition, any invalid state vector
\begin{align}
    \Pi_{\rm Init}\ket{\Phi^i} &= \frac{1}{\sqrt{T+1}}\Pi_{\rm Init}\ket{\phi_0^i} + \frac{1}{\sqrt{T+1}} \sum_{t=1}^{T}\Pi_{\rm Init}\ket{\phi_t^i} \\
    \nonumber
    &= \frac{1}{\sqrt{T+1}}\sum_{t=1}^{T} (\mathbbm{1}-H_{\rm Init})\ket{\phi_t^i} \\
    \nonumber
    &= \frac{1}{\sqrt{T+1}}\sum_{t=1}^{T}\ket{\phi_t^i},
    \nonumber
\end{align}
where the second equality follows since $\Pi_{\rm Init}$ is the projector onto the valid initial states and has a zero eigenvalue on the invalid ones. The last equality is true since $H_{\rm Init}$ only acts on the initial time step. Hence,
\begin{equation}
      \norm{\Pi_{\rm Init}\ket{\Phi^i}}^2 \leq T/(T+1) .
\end{equation}
Finally, note that $z + z^* = 2Re(z)$. To that end, note that $Re(z) \leq |z|$ since $Re(z)$ is the projection of $z$ on the real axis. Again, by noting that the nullspace of $\Pi_{\rm Init}$ are all the valid input states and since $\Pi_{\rm Init}$ is a projector. So it must be that $\Pi_{\rm Init}\ket{\Phi^v} = \ket{\Phi^v}$
\begin{equation}
    Re(\alpha^*\beta\bra{\Phi^v}\Pi_{\rm Init}\Pi_{\rm Final}\ket{\Phi^i}) \leq \abs{\alpha}\abs{\beta} |\bra{\Phi^v}\Pi_{\rm Final}\ket{\Phi^i}| \leq |\bra{\Phi^v}\Pi_{\rm Final}\ket{\Phi^i})|,
\end{equation}
where we note that $\Pi_{\rm Init}^{\dagger} = \Pi_{\rm Init}$ and that $\braket{\phi_t^v}{\phi_t^i}$ = 0 $\forall t$ since they are achieved by a unitary transform of $\braket{\phi_0^v}{\phi_0^i}$ = 0 because they are orthogonal in our construction. Now again note that $\Pi_{\rm Final}$ only acts on $\ket{\phi_t}$. Hence,
\begin{eqnarray}
    |\bra{\Phi^v}\Pi_{\rm Final}\ket{\Phi^i}| &=& \frac{1}{T+1} |\bra{\Phi_{T+1}^v}\Pi_{\rm Final}\ket{\Phi_{T+1}^i}| 
    \\
    \nonumber
    &\leq& \frac{1}{T+1}\norm{\Pi_{\rm Final}\ket{\Phi_{T+1}^v}}\\
    \nonumber
    &\leq& \frac{\sqrt{\varepsilon}}{T+1},
\end{eqnarray}
where the second to last inequality follows by Cauchy-Schwarz and the normalization of $\ket{\Phi_{T+1}^i}$ and the last inequality follows by the assumption that for the correct input which is not in the language, the acceptance probability is low: $\varepsilon$ by our assumption.

Finally, putting all the terms together, we get 
\begin{align}
     \norm{\Pi_{\rm Init}\Pi_{\rm Final}\ket{\Phi}}^2 &\leq \abs{\alpha}^2  \cdot \frac{T + \varepsilon}{T+1} + \abs{\beta}^2 \cdot \frac{T}{T+1} + \frac{2\sqrt{\varepsilon}}{T+1} \\
     \nonumber
     &= \frac{T (\abs{\alpha}^2 + \abs{\beta}^2) }{T+1} +  \varepsilon\left(\frac{\abs{\alpha}^2}{T+1} + \frac{2}{\sqrt{\varepsilon}(T+1)}\right)\\
     \nonumber
     &\leq \frac{T + \varepsilon + 2\sqrt{\varepsilon}}{T+1}\\
     &= 1 - \frac{c}{T+1},
\end{align}
where $c = 1- (\varepsilon + 2\sqrt{\varepsilon})$. Now, due to the identity $\cos^2(\theta/2) = \frac{1+\cos(\theta)}{2}$, we have 
\begin{equation}
    \cos^2(\theta/2) \leq \frac{1 + \sqrt{1-\frac{c}{T+1}}}{2} \leq 1 - \frac{c}{4(T+1)},
\end{equation}
where in the last inequality, we used $\sqrt{1-x} \leq 1-\frac{x}{2}$ for $x \in [0, 1]$. Finally, $\sin^2(\theta/2) + \cos^2(\theta/2) = 1$ implies that $\sin^2(\theta/2) \geq \frac{c}{4(T+1)} = \Omega(1/T)$. Using the minimum non-zero eigenvalue bound on $H_{\rm Prop}$ for $\lambda$ in Lemma~\ref{lemma-geometric-kitaev}, we get that $\lambda_{\min}(A_1 + A_2) = \Omega(1/T^3)$.  Finally, we note that $1/T^3 - \varepsilon  = O(1/{\rm poly}(n))$ because $T = O({\rm poly}(n))$ and $\varepsilon$ is exponentially close to $0$. This completes the soundness proof.
\end{proof}

Note that we did not need to separately bound $H_{\rm Penalty}$ here since it has a minimum eigenvalue of 1. This is because the only other type of states are illegal states and all of them are locally checkable. Moreover, locally checkable states form the basis of $H_{\rm Penalty}$. Hence if we are in the subspace of illegal states, it must be the superposition over basis states of $H_{\rm Penalty}$, each with eigenvalue 1.

\section{MA-hardness of 2-local 19-state stoquastic 
Hamiltonians on a line}

In this section, we prove the MA-hardness of stoquastic Hamiltonian problems on a line that builds upon many concepts introduced in the two-dimensional  construction in \cref{sec:2d-ma-hardness}, which we expect the reader to be familiar with.
In the one-dimensional case, we are faced with the same challenge as we faced in the two-dimensional case: to implement three-local Toffolis in a two-local way. Luckily, our technique of trading the Hilbert space dimension of the particles with the locality of the Hamiltonian works here as well.

Our clock construction is largely inspired by 
the ideas of Ref.\ \cite{Aharonov_2009} which 
has proven the QMA-completeness for the local Hamiltonian problem on a one-dimensional line with 12 state particles, except that our construction generalizes to single particles encoding two qubits and reduces the number of steps to reset a block by 
a constant factor because we do not use the turn flag on both the left and the right end, unlike the results of Ref.\ \cite{Aharonov_2009}. The high-level idea for a one-dimensional clock construction is to essentially stretch the two-dimensional construction as a one straight line, composed of `blocks', which is a collection of neighbouring particles. Each block captures a column of the 2D construction, or in other words, one layer of the the verification circuit.

First, let us imagine a downward phase in a given layer of our circuit, akin to the 2D construction. We work with a so-called ``active particle'' (see \cref{tab:active&inactive-states}) at all times. This particle uniquely determines where we are in the circuit. When we implement the a Toffoli gate $U$ in the 1D line, we will have the following mapping
\begin{equation}
    U\left(\ket{\dstatecc} \ket{\statebhboxed}\right)_{i, i+1} = \left(\ket{\statehrlboxed}\ket{\dstatecc}\right)_{i, i+1}
\end{equation}
on two sites (3 qubits) since $\dstatecc$ encodes two qubits and $\statebhboxed$ and $\statehrlboxed$ encode a single qubit each. On a one-dimensional line, we keep applying $U$ that propagates the state $\dstatecc$ from the left-most position in the block to the right-most one such that the block is in a state
\begin{equation}\label{eq-1D-block-end}
   \bigg\vert \statehrlboxed \statehrlboxed \statehrlboxed \dstatecc \bigg\vert,
\end{equation}
where the two lines on the boundary represent a block boundary, which is a position on the 1D line that is known a priori. This concludes what was the `downward' phase in our two-dimensional  construction (or the full layer of unitary gate application in the original MA-verification circuit). Now in order to go to the next layer, correspondingly to the next block on our one-dimensional 
grid, we need to reinitialize the next block as
\begin{equation}\label{eq-1D-block-start}
 \bigg\vert \dstatecc \statebhboxed \statebhboxed \statebhboxed \bigg\vert.    
\end{equation}

One can immediately notice that directly going from Eq.~\eqref{eq-1D-block-end} to Eq.~\eqref{eq-1D-block-start} involves interactions between left and right flags separated by $\dstatecc$. In the worst case, the first $\statehrlboxed$ in Eq.~\eqref{eq-1D-block-end} interacts with the first $\statebhboxed$ in Eq.~\eqref{eq-1D-block-start}. This is highly non-local. Thus, once the `downward' phase is finished and we reach the end of a block, we will `move' the state in the current block to the next one qudit by qudit: we first move the first particle from the current block all the way through the other particles, and place it in the first position of the next block, reinitializing it. Then we repeat this with each one of the particles as shown in steps 4 - 22 in \cref{fig:1D-propagation}. By the end of this process, the whole state is reinitialized in Eq.~\eqref{eq-1D-block-start} and we can proceed with the next layer of gates.

\subsection{One-dimensional transition rules and propagation terms}\label{subsec:transition-rules}
We first define the rules that govern the transition of states from one to another in our construction. Then  we demonstrate a  full cycle on 4 qubits  starting from a gate on one block up until reinitializing in the second block. We define a block as a sub-chain of length $n'-1$, where $n'$ is the number of qubits in the MA$_{\text{q}}$ verification circuit as described in \cref{sec:layers}. There are a total of $R$ such blocks (where $R$ is the number of layers in the original MA verification circuit) constituting the entire one-dimensional chain of size $(n'-1)R$.

\begin{table}[ht]
    \centering
   \begin{tabular}{|l|l|}
  \hline
 Inactive states  & Active states \\ \hline
  \dstatebb: Stores 2-qubit data when $\dstatecc$ is inactive.
  & \dstatecc: Gate flag (tracks gate application). \\
   \statebhboxed: Qubits on the right of the active site.
   & \rflag: Right-moving flag.\\
  \statehrlboxed: Qubits on the left of the active site. &  \lflag: Left-moving flag.\\
  \stateaboxed: Unborn (to right of all qubits).
  & \turnflag: Turning flag (changes sweep direction). \\
  \statedboxed: Dead (to left of all qubits).  & \\
  \hline
\end{tabular}
    \caption{Classification of states into `active' and `inactive' states. Both forward and backward transition of states is governed primarily by the active state. The two-local Hamiltonian terms applied at each time step are either on the left or on the right (neighbouring inactive state) of the active state. Neighbouring inactive states together with the current active state determine the transition to the next configuration.}
    \label{tab:active&inactive-states}
\end{table}

\begin{remark}[Notation]
    Note that $\rflag$, while pointing in only one horizontal direction, still carries two-dimensional data. The right-pointing shape is chosen for clarity.
\end{remark}

In the case of left and right moving flags, it is the next particle in the same direction. In case of $\dstatecc$ it is always right. As for the turning flag, note that we do not need the turn flag on the right turn, as opposed to
Ref.\ \cite{Aharonov_2009}. The reason we do not need a $\turnflag$ flag on the right turn is because the incoming right flag is already two dimensional interacting with the particle on its right --- $\stateaboxed$ --- which is one-dimensional to convert into a $\lflag$ and $\statehrlboxed$, which is the same total dimension of three. However, in the left sweep the one-dimensional flag, $\lflag$ interacting with $\statedboxed$ (another one-dimensional flag) cannot give $\statedboxed$ and $\rflag$ (two-dimensional), 
because that would map two-dimensional input to three-dimensional output. Each active site in a given configuration has a unique forward and backward transition rule as can be noted in the propagation rules. All particles except $\turnflag$, $\stateaboxed$, and $\statedboxed$ are data-carrying particles --- they store qubit data. Data carrying particles denoted by two shapes in a box are four-dimensional, and those with only one shape in a box are two dimensional.

We now turn to discussing the one-dimensional propagation rules.

\begin{enumerate}

\item (Gate application) \dstatecc \statebhboxed $\rightarrow$ \statehrlboxed \dstatecc
.\item (Right turn) \dstatecc $\big\vert$ \stateaboxed $\rightarrow$ \lflag $\big\vert$ \dstatebb
.\item (Left sweep)   \statehrlboxed \lflag $\rightarrow$ \lflag \statebhboxed; \dstatebb $\big\vert$ \lflag $\rightarrow$ \lflag $\big\vert$ \dstatebb
.\item (Left turn)   \statedboxed $\big\Vert$  \lflag $\rightarrow$ \statedboxed $\big\Vert$ \turnflag; \turnflag \statebhboxed $\rightarrow$ \statedboxed \rflag
.\item (Right sweep)   \rflag  \statebhboxed $\rightarrow$ \statehrlboxed \rflag; \rflag $\big\vert$ \dstatebb $\rightarrow$ \dstatebb $\big\vert$  \rflag
.\item (Right turn)   \rflag  \stateaboxed $\rightarrow$ \lflag  \statebhboxed
.\item (New round)   \turnflag $\big\vert$  \dstatebb $\rightarrow$ \statedboxed $\big\vert$  \dstatecc .
\end{enumerate}

\begin{remark}[Notation]
    $\big\Vert$ denotes that the rule holds for both with and without a boundary. For example the
first rule in Rule 4 (left turn) holds for with and without boundary.
\end{remark}

Using these, we can define $H_{\rm Prop}$ analogously to the two-dimensional case as in Section \ref{sec:prophamilt-2D}. Note that these terms are still stoquastic because the terms with $-U_l$ are stoquastic due to $U_l$ being Toffolis. And all the remaining terms are projections or have non-positive off-diagonals.

\begin{figure}[ht!]
\begin{tabular}{ccc|c|c|c}
Rule 1: &$1.$&  \statedboxed &  \dstatecc \statebhboxed \statebhboxed &   \stateaboxed \stateaboxed \stateaboxed  & \stateaboxed .\\ 
Rule 1: &$2.$& \statedboxed & \statehrlboxed \dstatecc \statebhboxed & \stateaboxed \stateaboxed \stateaboxed  & \stateaboxed .\\ 
Rule 2: &$3.$&  \statedboxed &  \statehrlboxed \statehrlboxed \dstatecc &   \stateaboxed \stateaboxed \stateaboxed  & \stateaboxed .\\ 
Rule 3: &$4.$&  \statedboxed &  \statehrlboxed \statehrlboxed \lflag &   \dstatebb \stateaboxed \stateaboxed  & \stateaboxed .\\
Rule 3: &$5.$&  \statedboxed &  \statehrlboxed \lflag  \statebhboxed &   \dstatebb \stateaboxed \stateaboxed  & \stateaboxed .\\
Rule 4: &$6.$&  \statedboxed &  \lflag\statebhboxed  \statebhboxed &   \dstatebb \stateaboxed \stateaboxed  & \stateaboxed .\\
Rule 4: &$7.$&  \statedboxed &  \turnflag\statebhboxed  \statebhboxed &   \dstatebb \stateaboxed \stateaboxed  & \stateaboxed .\\
Rule 5: &$8.$&  \statedboxed &  \statedboxed \rflag  \statebhboxed &   \dstatebb \stateaboxed \stateaboxed  & \stateaboxed .\\ 
Rule 5: &$9.$&  \statedboxed &  \statedboxed \statehrlboxed  \rflag &   \dstatebb \stateaboxed \stateaboxed  & \stateaboxed .\\ 
Rule 6: &$10.$&  \statedboxed &  \statedboxed \statehrlboxed  \dstatebb &   \rflag \stateaboxed \stateaboxed  & \stateaboxed .\\
Rule 3: &$11.$&  \statedboxed &  \statedboxed \statehrlboxed  \dstatebb &   \lflag \statebhboxed \stateaboxed  & \stateaboxed .\\
Rule 3: &$12.$&  \statedboxed &  \statedboxed \statehrlboxed  \lflag &   \dstatebb \statebhboxed \stateaboxed  & \stateaboxed .\\
Rule 4: &$13.$&  \statedboxed &  \statedboxed \lflag \statebhboxed &   \dstatebb \statebhboxed \stateaboxed  & \stateaboxed .\\
Rule 4: &$14.$&  \statedboxed &  \statedboxed \turnflag \statebhboxed &   \dstatebb \statebhboxed \stateaboxed  & \stateaboxed .\\
Rule 5: &$15.$&  \statedboxed &  \statedboxed \statedboxed \rflag &   \dstatebb \statebhboxed \stateaboxed  & \stateaboxed .\\
Rule 5: &$16.$&  \statedboxed &  \statedboxed \statedboxed \dstatebb &   \rflag \statebhboxed \stateaboxed  & \stateaboxed .\\
Rule 6: &$17.$&  \statedboxed &  \statedboxed \statedboxed \dstatebb &   \statehrlboxed \rflag  \stateaboxed  & \stateaboxed .\\
Rule 3: &$18.$&  \statedboxed &  \statedboxed \statedboxed \dstatebb &   \statehrlboxed \lflag  \statebhboxed  & \stateaboxed .\\
Rule 3: &$19.$&  \statedboxed &  \statedboxed \statedboxed \dstatebb &   \lflag \statebhboxed \statebhboxed  & \stateaboxed .\\
Rule 4: &$20.$&  \statedboxed &  \statedboxed \statedboxed \lflag &   \dstatebb \statebhboxed \statebhboxed  & \stateaboxed .\\
Rule 7: &$21.$&  \statedboxed &  \statedboxed \statedboxed \turnflag &   \dstatebb \statebhboxed \statebhboxed  & \stateaboxed .\\
Rule 1: &$22.$&  \statedboxed &  \statedboxed \statedboxed \statedboxed &   \dstatecc \statebhboxed \statebhboxed  & \stateaboxed .
\end{tabular}
\caption{A cycle for the forward and the reinitialization phase for 4 qubits. It takes $(n'-1)(2n' -1)$ time steps for a 
single cycle, where $n'$ is 
the 
number of qubits in the verification circuit. 
In total, it takes 
$T = (2n'-1)(n'-1)(R-1) + (n'-1)$, where $R$ is the total number of layers in the verification circuit.} \label{fig:1D-propagation}
\end{figure}

Note that in the simplest case, one can translate the one-dimensional construction 
of Ref.\ \cite{Aharonov_2009} and simply translate the one qubit left and right flags to the two qubit ones (similar to the two-dimensional construction). This would require us to to have a four-dimensional left and right flags each, denoting whether the inactive state $\dstatebb$ is on the left or the right of the active qubit. Right away, this introduces another 8 dimensions, 4 for left and right each. However, we can do better by exploiting the structure of our one-dimensional grid. Note that $\dstatebb$ always stays at the boundary --- either on its left or on its right. Whenever $\dstatebb$ is on the left of the boundary, it is always on the left of the active particle and similarly for the right hand side. Thus by looking at the boundary on the left or the right of $\dstatebb$, we can determine whether $\dstatebb$ is on the left of the active particle or on its right --- thus acting as a left and right flag without needing to introduce extra 8 dimensions for the two flags. And whenever the state is $\dstatecc$, then it is the active site itself. While this construction saves on the dimensionality of the particles, one can trade this with block length by decreasing the current block length of $n'-1$ to $n'/2$ (same as in the two-dimensional case) and increasing the particle dimension of each particle. One might notice that this insight can potentially be used to reduce the dimension in the two-dimensional version of the problem. We leave this as a clue for possible future work. 

\subsection{Illegal configurations and penalty terms in one spatial dimension}

Like the two-dimensional case, there is a possibility that we might be in an illegal state that does not correspond to a valid propagation from a valid starting state. We give an exhaustive list of all such illegal states in Table~\ref{tab:1D-prohibited-states}. While the table of 
one-dimensional prohibited states looks involved, it can simply be constructed by considering the allowed states on the left and right of each states, both with and without boundary by following the propagation cycle. However, there are configurations that do not appear in the configuration cycle and should still be penalized, which are shown in the final two rows of Table~\ref{tab:1D-prohibited-states}. Since all such states are nearest-neighbour (2-local) states, they can be penalized by the penalty Hamiltonian $ H_{\rm Penalty} = \sum_{i \in {\rm Illegal}} \Pi_i$, where the sum is over all the prohibited states in Table~\ref{tab:1D-prohibited-states}.

Unlike the two-dimensional case, however, we shall encounter illegal configurations that cannot have a two-local check like the ones listed in Table~\ref{tab:1D-prohibited-states}, and therefore are not penalized by $H_{\rm Penalty}$ . This problem also appeared in the the proof of QMA-hardness of general local Hamiltonians on a line \cite{Aharonov_2009}, and it arises because resetting one block requires moving all the states one-by-one from one block to the next. Let us first define the state configurations that appear in the 1D construction and prove some essential facts about them.

\begin{definition}[Qudit strings]\label{def-qubit-string}
For all valid configurations on the one-dimensional grid let the (\textbf{qudits}) strings be defined as any configuration of shapes between dead $\left(\statedboxed\right)$ and unborn $\left(\stateaboxed\right)$ states 
    \begin{equation}
        \statedboxed \cdots \statedboxed(\textbf{qudits})\stateaboxed \cdots
    \stateaboxed,
    \end{equation}
where the number of $\statedboxed$ and $\stateaboxed$ is not necessarily the same on the left and right, respectively.

\end{definition}

\begin{fact}[Legal qudit length]\label{fact:legal-qubit-length}
   All  legal (\textbf{qudits}) strings where the active state is either $\dstatecc$ or $\rflag$ have length $n'-1$. Moreover, 
     all  legal (\textbf{qudits}) strings where the active state is $\lflag$ or $\turnflag$ have length $n'$. 
\end{fact}
\begin{proof}
    Firstly, recall that the initial length of a block (from one boundary to the next) is $n'-1$, where $n'$ is the total number of qubits in the verification circuit. This is because $\dstatecc$ is four-dimensional while all the other $\statebhboxed$ states hold one qubit data. This holds true not just for the initial time step, but for all the time steps when $\dstatecc$ is the active state. Secondly, observe that whenever $\rflag$ is the active state, all the shapes in the (\textbf{qudits}) string are the states that carry qubit data. And since the number of qubits is $n'$, we can only ever have the block length of size $n'-1$ because whenever $\rflag$ is the active state, $\dstatebb$ carries two qubit data.
    The left flag, $\lflag$, and the turn flag, $\turnflag$, carry non-qubit data, and that has to be compensated with another qubit carrying state/shape. Thus, increasing the block length by 1.
\end{proof}

Observe, for example, that we can ensure that the $\turnflag$ is always at the end of the chain by a 2-local check shown in the third-to-last row of Table~\ref{tab:1D-prohibited-states}. This is done by only ever allowing $\statedboxed$ to the left of $\turnflag$ (both with and without boundary). 
Note, however, that we may have configurations that can't be checked by local projections. For example, configurations having more than the expected length of a qudit string, where the qudit string is the configuration of shapes with only dead states on its left and unborn states on its right. Consider having $\statebhboxed \big\vert \dstatebb$. This could happen in either first or the second step  (when the active site is $\dstatecc$), thus leading to the qudit string of length $n'$ on the grid. A qudit string of length $n'$ shall correspond to $n'+1$ qubits in the original verification circuit. And this is a contradiction since the original circuit had exactly $n'$ qubits. To remedy this, one might imagine directly penalizing $\statebhboxed \hspace{1.5mm} \big\vert \hspace{1.5mm}\dstatebb$. But note that this configuration is a legal one in other steps. For example, in Steps $5-8$. To penalize such illegal configurations with incorrect length of the qudit string, one needs to count the number of sites not having a $\statedboxed$ or a $\stateaboxed$. And it is impossible to count locally (without any memory). We will now show that such configurations that do not admit a local check are exactly those violating the legal lengths of qudit strings in Fact~\ref{fact:legal-qubit-length}.

\begin{claim} [Exceptions]\label{claim:exceptions}
    The only illegal configurations not penalized by the $H_{\rm Penalty}$
    are the (\textbf{qudits}) string violating length of a legal  (\textbf{qudits}) string in Fact~\ref{fact:legal-qubit-length}.
\end{claim}

\begin{proof}
While Fact~\ref{fact:legal-qubit-length} dictates the correct lengths for a given active state, it could also occur that despite having the correct length, we may still have an illegal configuration by being misaligned at the block boundaries but the rules for $\rflag$ and $\dstatecc$ in Table~\ref{tab:1D-prohibited-states} ensure that this does not happen. For the remaining two active states, the boundary alignment problem is non-existent as the the chain can never be aligned with a block boundary. This is because when the active state is either $\lflag$ or $\turnflag$, the length of any  (\textbf{qudits}) string is always $n'$, while the length of the chain between two boundaries in $n'-1$. Hence, the (\textbf{qudits}) string with the remaining two active states will always be bigger than the length of the block boundary.

Now observe that each state (active or not) has prohibited states on its left and right, both with and without boundary in Table~\ref{tab:1D-prohibited-states}. Hence, any illegal state that occurs next to any (active or non-active) state can be ruled out by a local check. Thus, the only states that remain are therefore exactly the states in Fact~\ref{fact:legal-qubit-length}.
\end{proof}

\begin{fact}[Single active site]\label{fact:single-active-state}
    Any (\textbf{qudits}) string not penalized by $H_{\rm Penalty}$ has exactly one active site.
\end{fact}
\begin{proof}
    Note that the legal states do not evolve to illegal ones and vice-versa. Legal states only evolve to legal states because the transition rules are unique and only map legal states to legal states. To evolve one active site in time or into another active site, it is necessary to have a two-local legal configuration for the (unique) propagation rule to apply to the concerned site. Every propagation rule either changes the position of the active site or delete the current active site in the qudit string by replacing it with a different active site. Hence, preserving the number of the active sites.
\end{proof}

  We are now ready to prove \Cref{lemma:illegal-non-local} which ensures that $H_{\rm Penalty} + H_{\rm Prop}$ have large energy in the subspace of non-locally checkable states, which shall be useful for proving soundness in the proof of \Cref{theorem:1D-MA-hardness}.

\begin{lemma}[Subspace of illegal states]\label{lemma:illegal-non-local}
    The minimum eigenvalue of $H_{\rm Penalty} + H_{\rm Prop}$ in the subspace of illegal states without a local check as listed in Claim~\ref{claim:exceptions} is $\Omega(1/T^3)$, where $T$ is the total number of steps in our one-dimensional construction.
\end{lemma}

\begin{proof}

    Firstly, note that the states in Claim~\ref{claim:exceptions} are the states without a local check since all the others are exhaustively ruled out by our local rules. By Fact~\ref{fact:single-active-state}, there is always a single, unique active state even in the (\textbf{qudits}) string of incorrect length. For every active state, depending on the configuration, we either have a local penalty or a propagation rule. If we had a local penalty, then then it would have been penalized by the penalty Hamiltonian. Otherwise, we will have a configuration that fits one of our propagation rules. We shall proceed by considering two cases: when the active site is $\dstatecc$; and when active site is not $\dstatecc$. 
    
    If the active site is $\dstatecc$, then in at most $n'-1$ steps it will either hit the left or right block boundary by backward and forward propagation rules. Consequently, it will be locally penalized because only $\statedboxed$ and $\stateaboxed$ are allowed across the left and right boundary, respectively. Thus, the number of steps to evolve to a locally detectable state from any active state are at most  $O(n'-1)$. Moreover, the length of the (\textbf{qudits}) string could be at most $(n'-1)\cdot R$, where $R$ is the number of rounds/layers in our modified circuit. Thus, we encounter at least one locally detectable state in $O((n')^2\cdot R)$. Thus the fraction of illegal states not admitting a local check is $\Omega(1/(n')^2\cdot R)$ when the the active state is $\dstatecc$.

    If the active state is any state other than $\dstatecc$, we claim that the same scaling holds. To see this, write the length of the (\textbf{qudits}) string, $N$, as $N = (n'-1)m + p$, where $0 \leq p \leq n'-1$ and $1 \leq m \leq R$. Regardless of the active state, the propagation rules guarantee that we reach $\dstatecc$ in $O(N\cdot (n'-1)) + O(N\cdot (n'-1)) + O(1)$, where each summand denotes the number of time steps for each of the $\lflag, \rflag$ and $\turnflag$, respectively, to change into the next active state. 
    Therefore, it takes a maximum of $O((n')^2R)$ steps to reach $\dstatecc$ and another $O(n')$ from there to reach a locally detectable state at the block boundary. This implies that we have at least $\Omega(1/((n')^2R)$ states that are locally checkable. As can be seen in Fig.~\ref{fig:1D-propagation}, the total number of steps $T = O((n')^2R)$, implying that the fraction of locally checkable illegal states is $\Omega(1/T)$, which also lower bounds $\sin^2(\theta/2)$ in Lemma~\ref{lemma-geometric-kitaev}. 
Finally, plugging in the bound on the minimum non-zero eigenvalue of $H_{\rm Prop} = \Omega(1/T^2)$ from 
    Ref.\ \cite{KitaevClassicalQuantumComputation} in Lemma~\ref{lemma-geometric-kitaev}, we get that 
    \begin{equation}\lambda_{\min}(H_{\rm Prop} + H_{\rm Penalty}) = \Omega(1/T^3),
    \end{equation}
    which proves the claim.
\end{proof}
 Note that we have already taken care of the cases when the states are illegal --- both with and without a local check. The only remaining case is that of the legal states. At this point, the proof of hardness is straightforward.

\subsection{MA-hardness proof}
 The full Hamiltonian $H = H_{\rm Init} + H_{\rm Prop}/2 + H_{\rm Penalty} + H_{\rm Final}$ in one-dimension is defined similarly to the 2D case.  $H_{\rm Init}$ and $H_{\rm Final}$ are also constructed similarly and are again stoquastic since they only comprise of projections.

\begin{proof}[Proof of Theorem~\ref{theorem:1D-MA-hardness}] 
First, note that with  $H_{\rm Init}$ and $H_{\rm Final}$ analogously as in Theorem~\ref{theorem:2d-ma-hard} the completeness analysis follows directly from \Cref{subsection:completeness-proof}. For soundness, we will argue that regardless of the type of subspace, legal or illegal (locally checkable or not), the minimum eigenvalue of  $H = H_{\rm Init} + H_{\rm Prop}/2 + H_{\rm Penalty} + H_{\rm Final}$ is at least $1/T^3$, where $T$ is the total number of time steps in the verification circuit. Firstly, if we are in a locally checkable configuration, then $\lambda_{\min}(H_{\rm Penalty}) \geq 1$ since $H_{\rm Penalty}$ is diagonal in the basis of locally checkable illegal states. Since all other Hamiltonian terms are positive semi-definite, the $\Omega(1/T^3)$ bound holds. For the legal states, the bound still holds as shown in Section~\ref{subsection:soundness-proof}. For non-locally checkable configurations, thanks to \Cref{lemma:illegal-non-local}, the $\Omega(1/T^3)$ continues to hold. Finally, the particles in the one-dimensional case are 19-dimensional because $\stateaboxed$, $\statedboxed$ and $\turnflag$ are one-dimensional each. $\dstatecc$ and $\dstatebb$ are four-dimensional each, while the remaining four particles are two-dimensional each. The total dimension is simply the sum of each dimension since the particle Hilbert space is a direct sum of each of these configurations.
\end{proof}

\section{StoqMA-completeness}\label{section:stoqma-completeness}

Bravyi, Bessen and Terhal \cite{bravyi2006merlinarthur} have introduced a new class, referred to as
StoqMA, and showed that stoquastic LH-MIN is a complete problem for the class StoqMA. Recall that StoqMA is a class that allows the verification circuit to only have classically reversible gates followed by measurements only in the $\{\ket{+}, \ket{-}\}$ basis. This basis restriction for the measurement does not allow for the exponential amplification/suppression of completeness/soundness parameters. To ensure an inverse polynomial gap between the eigenvalues of the Hamiltonian corresponding to yes and no instances, the authors in Ref.\ \cite{bravyi2006merlinarthur} used first order time-independent perturbation theory. They divide the full Hamiltonian into the unperturbed part and the perturbation Hamiltonian  --- which they take to be $H_{\rm Final}$ \footnote{The construction in Ref.\ \cite{bravyi2006merlinarthur} does not have the $H_{\rm Clock}$ term. Nevertheless, the same lower bound holds even with the $H_{\rm Clock}$. See Lemma 3.11 in
Ref.\ \cite{aharonov2008adiabatic}.} . To ensure that this ``perturbation"  was small, as is required for the perturbative approach to work, they scaled $H_{\rm Final}$ by a parameter $\delta \ll \Omega(1/T^3)$, where $\Omega(1/T^3)$ is the lower bound on the spectral gap \footnote{The spectral gap is the difference between the smallest and second smallest eigenvalue.} of 

\begin{equation}
A_1 \coloneqq H_{\rm Init} + H_{\rm Prop}.
\end{equation}
The spectral analysis is borrowed from Lemma 3.11 \cite{aharonov2008adiabatic} by setting $s=1$ therein. 

\begin{corollary}[StoqMA-completeness]\label{corollary-stoqma-completeness}
     Geometrically two-local stoquastic LH-MIN is StoqMA-complete in two and one dimensions with 14 and 19 state particles, respectively.
\end{corollary}
\begin{proof}
    The proof almost exactly follows from Appendix. 3 in Ref.\ \cite{bravyi2006merlinarthur}. Let $A_1 \coloneqq H_{\rm Init} + H_{\rm Prop} + H_{\rm Clock}$ and the perturbation Hamiltonian to be $A_2 = \delta H_{\rm Out}$ where $\delta \ll 1/T^3$. Now, compared to the setting of Ref.\ \cite{bravyi2006merlinarthur}, we only have an additional $H_{\rm Clock}$ term in the ``unperturbed'' Hamiltonian, $A_1$, that has the valid history state as its zero eigenstate.
    So the completeness analysis in Ref.\ \cite{bravyi2006merlinarthur} still holds. Soundness also simply follows since $\bra{\psi}H_{\rm Clock}\ket{\psi} \geq 0$ for any quantum state vector $\ket{\psi}$. The desired statement follows by noting that these two properties of $H_{\rm Clock}$ hold in both our 2D and one-dimensional constructions. 
\end{proof}

\section{Conclusion and future work}

To understand the computational complexity
of physically meaningful 
families of Hamiltonians is one of the core aims of the field of
Hamiltonian complexity, located at the intersection of 
theoretical computer science and quantum many-body physics. 
In this work, we studied the hardness of geometric stoquastic local Hamiltonian problem in two natural readings, for a two-dimensional grid and on a one-dimensional line. For
both settings, we prove that it remains MA-hard if the particles have sufficiently high local dimension. From a Physics perspective, this may be seen as a surprising result: After all, such systems are substantially closer to the reality of physically plausible sign-free Hamiltonians, and hence 
this maintained hardness may come as a surprise. Also, the sign-free geometrically local 
Hamiltonians are commonly considered the ``easy'' 
instances of Hamiltonians from the perspective of Monte Carlo sampling, so viewed from this perspective, it may be surprising to see that the approximation 
of their ground state energy is MA-hard. It may be worth noting that the general question of the computational complexity of bringing geometrically local Hamiltonians into
a manifestly stoquastic form is still open.

We think that one can further reduce the dimension of the particles in the two-dimensional construction by using the ideas from the one-dimensional construction, and hence
bring the families of systems considered closer to plausible models
from the condensed matter physics context: Only having the primary working qubit (the active state) to be of high enough dimension. Alternatively, only applying two qubit gates would also significantly reduce particle dimension. For that to work, it would be imperative to find two qubit gates that are classically reversible, universal  and preserve stoquasticity in the construction. 

It would also be interesting to
bring this type of study into contact with efforts of understanding
the computational complexity of
two-local qubit Hamiltonians
\cite{Cubitt}. Indeed, 
arguably, the physically best motivated and theoretically interesting is the case where the particle dimension is $d=2$, i.e., the qubit case. We expect that stoquastic perturbation gadgets would be useful to address the qubit case. Another
presumably ``easier'' version of the problem is the translationally invariant case \cite{PhysRevA.76.030307}, where all the local Hamiltonian terms are the same on all sites. It would be exciting to see if the problem persists to be MA$_{\text{EXP}}$-hard \footnote{In the translationally invariant setting, the input is the number of particles $n$, which requires $\log n$ bits to specify. Thus any polynomial-in-$n$ algorithm is exponential in the input size.} in the translationally invariant setting, again closer to physical reality. One could also think of the complexity of the 
variant of the problem in which 
certain qubits are pinned \cite{PhysRevA.103.012604},
where pinning refers to fixing the state of a small subsystem.
On a higher level, it is important to bring results in Hamiltonian complexity closer to the physical reality of natural quantum many-body systems. It is the hope that the present work that takes new  
steps in this direction stimulates further such efforts along these lines.

\section{Acknowledgements}
We are grateful to the anonymous reviewers for constructive feedback that improved the presentation of this work. We further thank the ERC (DebuQC), the QuantERA (HQCC), and the DFG (CRC 183), 
Berlin Quantum, the Munich Quantum Valley, and the BMFTR
(QuSol, Hybrid++) for support.
ABG is supported by ANR
JCJC TCS-NISQ ANR-22-CE47-0004 and by the PEPR integrated
project EPiQ ANR-22-PETQ-0007 part of Plan France 2030. This work was done in part while ABG was
visiting the Simons Institute for the Theory of Computing. Part of this work was done while AR was visiting ABG at LIP6, Sorbonne Universit{\'e}/CNRS funded by ANR
JCJC TCS-NISQ ANR-22-CE47-0004.

\printbibliography[heading=bibintoc]

\newpage
\begin{appendix}

\section{Two-dimensional construction: Illegal clock states}\label{sec-illegal-states}

\begin{table}[H]
\begin{tabular}{|l|l|}
\hline
Forbidden  & Guarantees \\ \hline \hline
   $\fstateelse \dstated$ & Any state that is not 
   dead should always be on the right of dead states. \\ \hline
   $\stateda \fstateelse$ & Unborn states should always be on the right of all the other states. \\ \hline
   $\dstated \stateda$ & Dead and alive states must not be horizontally adjacent. \\ 
   \hline
   $\dstatecc \dstatecc$, $\dstatecb \dstatecb$,&\\
   $\dstatebb \dstatebb$, $\dstatebb \dstatecb$,&\\
   $\dstatebb \dstatecc$, $\dstatecb \dstatebb$,&\\
   $\dstatecb \dstatecc$, $\dstatecc \dstatebb$,&\\
   $\dstatecc \dstatecb$ & Computational "substate" states must not be horizontally adjacent. \\
   \hline
   $\ontop{\fstateelse}{\stateda}$, $\ontop{\stateda}{\fstateelse}$ & Only $\stateda$ should be above and below itself.
    \\
   \hline
   $\ontop{\dstated}{\fstateelse}$ & Only $\dstated$ can be below itself.\\ \hline
   $\ontop{\dstatebb}{\dstated}$, $\ontop{\dstatecb}{\dstated}$ & Only $\dstatecc$ or $\dstated$ can be above $\dstated$.\\ \hline
   $\ontop{\dstatebb}{\fstateelse}$ & Only $\dstatebb$ can be below itself.  \\ \hline
    $\ontop{\dstated}{\dstatebb}$, $\ontop{\dstatecc}{\dstatebb}$,
   $\ontop{\dstatecb}{\dstatebb}$ &  $\dstatebb$ either allows itself or $\stateda$ in its immediate upper row. \\ 
   
   \hline
   $\ontop{\stateda}{\dstatecc}$, $\ontop{\dstated}{\dstatecc}$, $\ontop{\dstatebb}{\dstatecc}$,&\\
    $\ontop{\dstatecb}{\dstatecc}$, $\ontop{\stateda}{\dstatecb}$, $\ontop{\dstated}{\dstatecb}$,&\\
    $\ontop{\dstatebb}{\dstatecb}$, $\ontop{\dstatecb}{\dstatecb}$
  & Only $\dstatecc$  occurs above $\dstatecc$ and/or $\dstatecb$.\\
   \hline
\end{tabular}
    \caption{Table of illegal states on a two-dimensional grid between neighbouring particles that can be penalized using only two-local Hamiltonian terms. Here, $\protect\fstateelse$ is a placeholder for all states except the listed neighbouring state.}
    \label{tab-inter-forbidden}
\end{table}

\vspace*{-6mm}

\section{Illegal states in the one-dimensional construction}\label{sec-
illegal-states-1D}
\begin{table}[H]
\begin{tabular}{|l|l|}
\hline
Forbidden  & Where $\stateelseboxed$ is \\ \hline \hline
   $\stateaboxed \hspace{1mm} \big\Vert \hspace{1mm}\stateelseboxed$ & Anything except $\stateaboxed$. \\ \hline
   $\stateelseboxed \hspace{1mm} \big\Vert \hspace{1mm}\statedboxed$ & Anything except $\statedboxed$. \\ 
   \hline
   $\statebhboxed \hspace{4mm}\stateelseboxed$ & Anything except $\statebhboxed$.  \\ 
   \hline
   $\statebhboxed \hspace{1.5mm} \big\vert \hspace{1.5mm}\stateelseboxed$ & Anything except $\dstatebb$ or $\stateaboxed$.  \\ 
   \hline
  $\stateelseboxed
  \hspace{4mm}\statebhboxed $ & Anything except $\statebhboxed, \dstatecc, \dstatebb, \lflag, \rflag,$ or $\turnflag$.  \\ 
   \hline
   $\statehrlboxed
  \hspace{4mm}\stateelseboxed $ & Anything except $\dstatecc, \dstatebb, \lflag,$ or $\rflag$. \\ 
   \hline
  $\stateelseboxed
  \hspace{4mm}\statehrlboxed $ & Anything except $\statehrlboxed$.  \\ 
   \hline
  $\stateelseboxed
  \hspace{1.5mm} \big\vert \hspace{1.5mm}\statehrlboxed $ & Anything except $\dstatebb$ or $\statedboxed$.  \\ 
   \hline
    $\dstatebb
  \hspace{4mm}\stateelseboxed $ & Anything except $\stateaboxed$ or $\statebhboxed$.  \\ 
  \hline
  $\dstatebb
  \hspace{1.5mm} \big\vert \hspace{1.5mm}\stateelseboxed $ & Anything except $\statehrlboxed, \rflag$, or  $\lflag$.  \\ 
  \hline
  $\stateelseboxed
  \hspace{4mm}\dstatebb $ & Anything except $\statehrlboxed$, or  $\statedboxed$.  \\ 
  \hline
   $\stateelseboxed
  \hspace{1.5mm} \big\vert \hspace{1.5mm}\dstatebb $ & Anything except $\statebhboxed, \rflag, \lflag$, or  $\turnflag$.  \\ 
  \hline
  $\dstatecc
  \hspace{4mm}\stateelseboxed $ & Anything except $\statebhboxed$.  \\ 
  \hline
  $\dstatecc
  \hspace{1.5mm} \big\vert \hspace{1.5mm}\stateelseboxed $ & Anything except $\stateaboxed$.  \\ 
  \hline
   $\stateelseboxed
  \hspace{4mm} \dstatecc$ & Anything except $\statehrlboxed$.  \\
  \hline
   $\stateelseboxed
   \hspace{1.5mm} \big\vert \hspace{1.5mm} \dstatecc$ & Anything except $\statedboxed$.  \\
   \hline
    $\rflag
  \hspace{4mm}\stateelseboxed $ & Anything except $\statebhboxed$, or $\stateaboxed$.  \\ 
  \hline
    $\rflag
  \hspace{1.5mm} \big\vert \hspace{1.5mm}\stateelseboxed$, $\stateelseboxed
  \hspace{1.5mm} \big\vert \hspace{1.5mm}\rflag $  & Anything except $\dstatebb$. \\ 
  \hline
   $\stateelseboxed
  \hspace{4mm}\rflag $ & Anything except $\statehrlboxed$, or $\statedboxed$. \\ 
  \hline
   $\lflag
  \hspace{4mm}\stateelseboxed $ & Anything except $\statebhboxed$.  \\
  \hline
   $\lflag
    \hspace{1.5mm} \big\vert \hspace{1.5mm}\stateelseboxed $ & Anything except $\dstatebb$.  \\
  \hline
  $\stateelseboxed
  \hspace{4mm}\lflag $ & Anything except $\statehrlboxed$, or $\statedboxed$.  \\
  \hline
  $\stateelseboxed
  \hspace{1.5mm} \big\vert \hspace{1.5mm}\lflag $ & Anything except $\dstatebb$, or $\statedboxed$.  \\
  \hline
   $\turnflag
  \hspace{4mm}\stateelseboxed $ & Anything except $\statebhboxed$.  \\
  \hline
  $\turnflag
    \hspace{1.5mm} \big\vert \hspace{1.5mm}\stateelseboxed $ & Anything except $\dstatebb$.  \\
  \hline
  $\stateelseboxed
    \hspace{4mm}\turnflag$, $\stateelseboxed
    \hspace{1.5mm} \big\vert \hspace{1.5mm}\turnflag$ & Anything except $\statedboxed$.  \\
  \hline
  $\stateelseboxed
   \hspace{1.5mm} \big\vert \hspace{1.5mm}\statebhboxed$ & Any state. \\
   \hline
   $\statehrlboxed
  \hspace{1.5mm} \big\vert \hspace{1.5mm}\stateelseboxed$ & Any state. \\
  \hline
   
\end{tabular}
 \caption{Prohibited states in one spatial dimension.}
    \label{tab:1D-prohibited-states}
\end{table}

\end{appendix}
\end{document}